\documentclass[10pt]{amsart}
\usepackage{color}
\usepackage{pgf}
\usepackage{mathtools}
\usepackage{enumerate}
\usepackage{tikz}

\usetikzlibrary{patterns}
\usetikzlibrary{arrows,shapes,positioning}
\usetikzlibrary{decorations.markings}
\usetikzlibrary[matrix] 

\tikzstyle directed=[postaction={decorate,decoration={markings,
    mark=at position .65 with {\arrow[arrowstyle]{latex}}}}]

\tikzstyle arrowstyle=[scale=1]

\usepackage[left=3.cm, right=3.cm, top=3.5cm, bottom=3.5cm]{geometry}
\usepackage{amssymb}
\usepackage{cite}

\newtheorem{thm}{Theorem}[section]
\newtheorem{cor}[thm]{Corollary}

\newtheorem{prop}[thm]{Proposition}

\theoremstyle{definition}

\newtheorem{example}[thm]{Example}
\newtheorem{definition}[thm]{Definition}
\newtheorem{remark}[thm]{Remark}
\numberwithin{equation}{section}
\numberwithin{figure}{section}

\newcommand{\Wr}{\operatorname{Wr}}
\newcommand{\N}{\mathbb{N}}

\newcommand{\Z}{\mathbb{Z}}
\newcommand{\C}{\mathbb{C}}
\newcommand{\R}{\mathbb{R}}
\newcommand{\bbeta}{{\boldsymbol{\beta}}}
\newcommand{\bmu}{{\boldsymbol{\mu}}}
\newcommand{\bpi}{{\boldsymbol{\pi}}}

\newcommand{\supth}{{}^{\rm{th}}}

\renewcommand{\H}{H}
\newcommand{\cL}{L}
\newcommand{\cM}{\mathcal{M}}
\newcommand{\cZ}{\mathcal{Z}}

\newcommand{\ddz}{\frac{{\rm d}}{{\rm d}z}}

\renewcommand{\th}{\widetilde{H}}

\newcommand{\h}{H}
\newcommand{\Pfour}{{\rm P_{IV}}}
\newcommand{\Pfive}{{\rm P_{V}}}
\newcommand{\PI}{$\mathrm{P}_{\mathrm{I}}$}
\renewcommand{\PII}{$\mathrm{P}_{\mathrm{II}}$}
\newcommand{\PIII}{$\mathrm{P}_{\mathrm{III}}$}
\newcommand{\PIV}{$\mathrm{P}_{\mathrm{IV}}$}
\newcommand{\PV}{$\mathrm{P}_{\mathrm{V}}$}

\renewcommand{\boxdot}{{\ \clap{\raise0.25ex\hbox{$\bullet$}}\clap{$\square$}\ }}
\newcommand{\emptybox}{\hbox{$\square$}}

\newcommand{\comment}[1]{}
\newcommand{\p}{Painlev\'{e}}
\newcommand{\ep}{\varepsilon}
\newcommand{\bk}{B\"acklund}

\def\d{{\rm d}}
\newcommand{\deriv}[3][]{\frac{\d^{#1}{#2}}{{\d{#3}}^{#1}}}
\renewcommand{\a}{\alpha}
\newcommand{\ex}[1]{\exp(#1)}

\newcommand{\cc}[1]{c_{#1}}

\newcommand{\VY}{Yablonskii--Vorob'ev polynomials}

\newcommand{\odes}{ordinary differential equations}

\def\Anfig#1#2{\includegraphics[width=#1in]{figures/#2}}

\begin{document}

\title[Cyclic Maya diagrams]{Cyclic Maya diagrams and rational solutions of higher order Painlev\'{e} systems}

\author{Peter A. Clarkson}
\address{School of Mathematics, Statistics and Actuarial Science,
University of Kent, CT2 7NF, UK.}

\author{David G\'omez-Ullate}
\address{Escuela Superior de Ingenier\'ia, Universidad de C\'adiz, 11519 Puerto Real, Spain.}
\address{Departamento de F\'isica Te\'orica, Universidad Complutense de
  Madrid, 28040 Madrid, Spain.}

\author{Yves Grandati}
\address{ Laboratoire de Physique et Chimie Th\'eoriques, Universit\'e de Lorraine, 57078 Metz, Cedex 3, France.}
\author{Robert Milson}
\address{Department of Mathematics and Statistics, Dalhousie University,
  Halifax, NS, B3H 3J5, Canada.}
\email{P.A.Clarkson@kent.ac.uk, david.gomezullate@uca.es,yves.grandati@univ-lorraine.fr,   rmilson@dal.ca}
\begin{abstract}
This paper focuses on the construction of rational solutions for the $A_{2n}$-\p\ system, also called the Noumi-Yamada system, which are considered the higher order generalizations of $\Pfour$. In this even case, we introduce a method to construct the rational solutions based on cyclic dressing chains of Schr\"{o}dinger operators with potentials in the class of rational extensions of the harmonic oscillator. Each potential in the chain can be indexed by a single Maya diagram and expressed in terms of a Wronskian determinant whose entries are Hermite polynomials. We introduce the notion of cyclic Maya diagrams and we characterize them for any possible period, using the concepts of genus and interlacing. The resulting classes of solutions can be expressed in terms of special polynomials that generalize the families of generalized Hermite, generalized Okamoto and Umemura polynomials, showing that they are particular cases of a larger family.

\end{abstract}

\maketitle
\section{Introduction}

The set of six nonlinear second order \p\ equations ${\rm P_I,\dots,P_{VI}}$ 
\comment{\begin{align}
\label{eq:P1}
\deriv[2]{u}{z}&=6u^2+z,\\
\label{eq:P2}
\deriv[2]{u}{z}&=2u^3+zu+\a,\\
\label{eq:P3}
\deriv[2]{u}{z}&=\frac{1}{u}\left(\deriv{u}{z}\right)^2-\frac{1}{z}\deriv{u}{z}
+\frac{\a u^2+\beta}{z}+\gamma u^3+\frac{\delta}{u},\\
\label{eq:P4}
\deriv[2]{u}{z}&=\frac{1}{2u}\left(\deriv{u}{z}\right)^2
+\frac32u^3+4zu^2+2(z^2-\a)u+\frac{\beta}{u},\\
\label{eq:P5}
\deriv[2]{u}{z}&=\left(\frac{1}{2u}+\frac{1}{u-1}\right)\left(\deriv{u}{z}\right)^2-\frac{1}{z}\deriv{u}{z}
+ \frac{(u-1)^2}{z^2}\left(\a u+\frac{\beta}{u}\right) +
\frac{\gamma u}{z}+\frac{\delta u(u+1)}{u-1},\\
\label{eq:P6}
\deriv[2]{u}{z}&= \frac{1}{2}\left(\frac{1}{u} +
\frac{1}{u-1} + \frac{1}{u-z}\right)\!\left(\deriv{u}{z}\right)^{2} - 
\left(\frac{1}{z} + \frac{1}{z-1} + \frac{1}{u-z}\right)\! \deriv{u}{z}\\
&\qquad\qquad+\frac{u(u-1)(u-z)}{z^2(z-1)^2}\left\{\a +
\frac{\beta z}{u^2} + \frac{\gamma (z-1)}{(u-1)^2} + 
\frac{\delta z(z-1)}{(u-z)^2}\right\},\nonumber
\end{align}
with $\a$, $\beta$, $\gamma$ and $\delta$ constants,}%
have been the focus of intense study from
many different angles in the past century \cite{refPAC06review,gromak2008painleve}.
Their defining property is that their solutions have no movable branch points, i.e.\ the locations of multi-valued singularities of any of the solutions are independent of the particular solution chosen and so are dependent only on the equation.
The \p\ equations, whose solutions are called \p\ transcendents, are now considered to be the nonlinear analogues of special functions, cf.~\cite{clarkson2003painleve,refPAC06review,refFIKNbook}. 
Although, in general, these functions are transcendental in the sense that they cannot be expressed in terms of previously known functions, the \p\ equations, except the first, also possess special families of solutions that can be expressed via rational functions, algebraic functions or the classical special functions, such as Airy, Bessel, parabolic cylinder, Whittaker or hypergeometric functions, for special values of the parameters, see, for example, \cite{refPAC06review,gromak2008painleve} and the references therein.

Rational solutions of the second \p\ equation (\PII) 
\begin{equation}\label{eq:P2}
u''=2u^3+zu+\a,
\end{equation}
with $\a$ an arbitrary constant and $'\equiv \d/\d z$,
were studied by Yablonskii \cite{yablonskii1959rational} and Vorob'ev \cite{vorob1965rational},
in terms of a special class of polynomials, now known as the \textit{\VY} $Q_n(z)$, which are polynomials of degree $\tfrac12n(n+1)$. 
Clarkson and Mansfield \cite{refCM03} investigated the locations of the roots of the \VY\ in the complex plane and observed that these roots have a very regular, approximately triangular structure; the term ``approximate" is used since the patterns are not exact triangles as the roots lie on arcs rather than straight lines.
Bertola and Bothner \cite{refBB15} and Buckingham and Miller \cite{refBM14,refBM15} studied the \VY\ $Q_n(z)$ in the limit as $n\to\infty$ and showed that the roots lie in a triangular region with elliptic sides which meet with interior angle $\tfrac25\pi$. 
Further Buckingham and Miller \cite{refBM14,refBM15} show that in the limit as $n\to\infty$, the rational solution of \PII\ tends to the \textit{tritronqu\'ee solution} of the first \p\ equation (\PI) 
\begin{equation}\label{eq:P1}
u''=6u^2+z.
\end{equation}

Okamoto \cite{okamoto1987studies3} obtained special polynomials associated with some of the rational solutions of the fourth \p\ equation (\PIV) 
\begin{equation}\label{eq:P4}
u''=\frac{(u')^2}{2u} 
+\frac32u^3+4zu^2+2(z^2-\a)u+\frac{\beta}{u},
\end{equation}
with $\a$ and $\beta$ constants, which are analogous to the \VY.
Noumi and Yamada \cite{noumi1999symmetries} generalized Okamoto's results and expressed all rational solutions of \PIV\ in terms of two types of special polynomials, now known as the \textit{generalized Hermite polynomials} $H_{m,n}(z)$ and \textit{generalized Okamoto polynomials} $Q_{m,n}(z)$, both of which are determinants of sequences of Hermite polynomials. The structure of the roots of these polynomials is studied in \cite{clarkson2003fourth}, where it is shown that the roots of the generalized Hermite polynomials have an approximate rectangular structure and the roots of the generalized Okamoto polynomials have a combination of approximate rectangular and triangular structures. 
Recent studies on the asymptotic distribution of zeros of certain generalised Hermite polynomials $H_{m,n}(z)$ as $m,n\to\infty$ are given in
\cite{refBuckPIV,refMR18,masoero2019roots} and of the Okamoto polynomials in \cite{NovokShchel14}.
Buckingham \cite{refBuckPIV} also obtained an explicit characterization of the boundary curve in the case of generalized Hermite polynomials.

Umemura \cite{refUmemura01} derived special polynomials associated with rational solutions of the third \p\ equation (\PIII) and the fifth \p\ equation (\PV) 
\comment{\begin{align}\label{eq:P3}
u''&=\frac{(u')^2}{u} - \frac{u'}{z}
+\frac{\a u^2+\beta}{z}+\gamma u^3+\frac{\delta}{u},\\
\label{eq:P5}
u''&=\left(\frac{1}{2u}+\frac{1}{u-1}\right)(u')^2 - \frac{u'}{z} + \frac{(u-1)^2}{z^2}\left(\a u+\frac{\beta}{u}\right) +
\frac{\gamma u}{z}+\frac{\delta u(u+1)}{u-1},
\end{align}
with $\a$, $\beta$, $\gamma$ and $\delta$ constants,}%
which are determinants of sequences of associated Laguerre polynomials. 

The special polynomials associated with rational solutions of the \p\ equations arise in several applications:
\begin{enumerate}[(i)]
\item the \VY\ arise in the transition behaviour for the semi-classical sine-Gordon equation \cite{refBM12},
in boundary value problems \cite{refBNRS},
in moving boundary problems \cite{refRogers15,refRogers16,refRogers17},
and in symmetry reductions of a Korteweg capillarity system \cite{refRC17} and cold plasma physics \cite{refRC18};
\item the generalized Hermite polynomials arise as multiple integrals in random matrix theory \cite{refFW01},
in supersymmetric quantum mechanics \cite{bermudez2012complex,bermudez2012complexb,marquette2016,Novokshenov18},
in the description of vortex dynamics with quadrupole background flow \cite{Clarkson2009vortices},
and as coefficients of recurrence relations for orthogonal polynomials \cite{chen2006painleve,clarkson2014relationship,refWVAbook};
\item the generalized Okamoto polynomials arise in supersymmetric quantum mechanics \cite{marquette2016} and 
generate rational-oscillatory solutions of the de-focusing nonlinear Schr{\"o}dinger equation \cite{clarkson2006special};
\item the Umemura polynomials arise as multivortex solutions of the complex sine-Gordon equation \cite{refBP98,refOB05}, and in 
Multiple-Input-Multiple-Output wireless communication systems \cite{refCCBC}.
\end{enumerate}
%

A very successful approach in the study of rational solutions to \p\ equations has been through the geometric methods developed by the Japanese school, most notably by Noumi and Yamada \cite{noumi2004painleve}. The core idea is to write the scalar equations as a system of  first order nonlinear equations. For instance, Noumi and Yamada \cite{noumi2004painleve} showed that \PIV\ \eqref{eq:P4}
is equivalent to the following system of three first order equations
\begin{align}\label{eq:P4system}
f_0' + f_0(f_1-f_2) &= \a_0, \nonumber\\
f_1' + f_1(f_2-f_0) &= \a_1,\\
f_2' + f_2(f_0-f_1) &= \a_2, \nonumber 
\end{align}
with $'\equiv \d/\d z$ and $\a_j$, $j=0,1,2$ constants,
subject to the normalization conditions
\begin{equation}\label{eq:P4normalization}
f_0+f_1+f_2=z,\qquad \a_0+\a_1+\a_2=1.
\end{equation}
Once this equivalence is shown, it is clear that the symmetric form of
$\Pfour$ \eqref{eq:P4system}, sometimes referred to as $\rm{s}\Pfour$,
is easier to analyse. In particular, Noumi and Yamada
\cite{noumi1999symmetries} showed that the system \eqref{eq:P4system}
possesses a symmetry group of \bk\ transformations acting on the
tuple of solutions and parameters
$(f_0,f_1,f_2 |\a_0,\a_1,\a_2)$. This symmetry group is
the affine Weyl group $A_2^{(1)}$, generated by the operators
$\{ \boldsymbol{\pi},\textbf{s}_0, \textbf{s}_1, \textbf{s}_2\}$ whose
action on the tuple $(f_0,f_1,f_2 |\a_0,\a_1,\a_2)$ is
given by:
\begin{align}\label{eq:BT}
&\textbf{s}_k (f_j)=f_j-\frac{\a_k\delta_{k+1,j}}{f_k}+\frac{\a_k\delta_{k-1,j}}{f_k}, \nonumber\\
&\textbf{s}_k(\a_j)=\a_j-2\a_j\delta_{k,j}+\a_k(\delta_{k+1,j}+\delta_{k-1,j}),\\
&\boldsymbol{\pi}(f_j)=f_{j+1},\qquad  \boldsymbol{\pi}(\a_j)=\a_{j+1} \nonumber
\end{align}
where $\delta_{k,j}$ is the Kronecker delta and $j,k=0,1,2 \mod(3)$.
The technique to generate rational solutions is to first identify a
number of very simple rational \textit{seed solutions}, and then
successively apply the \bk\ transformations \eqref{eq:BT} to
generate families of rational solutions.

This is a beautiful approach which makes use of the hidden group theoretic structure of transformations of the equations,
but the solutions built by dressing seed solutions are not very explicit,
in the sense that one needs to iterate a number of \bk\ transformations \eqref{eq:BT} on the functions and parameters in order to obtain the desired solutions. Questions such as determining the number of zeros or poles of a given solution constructed in this manner seem very difficult to address. For this reason, alternative representations of the rational solutions have also been investigated, most notably the determinantal representations \cite{kajiwara1996determinant,kajiwara1998determinant}, and representations in terms of Schur polynomials,\cite{noumi1999symmetries} and universal characters \cite{tsuda2005universal}.

The system of first order equations \eqref{eq:P4system} admits a
natural generalization to any number of equations, and it is known as
the $A_{N}$-\p\ or the Noumi-Yamada system. The $A_{2n}$-\p\ system  is considerably simpler (for reasons that will be
explained later), and it is the one we will focus on this paper. In
this case, the system has the form:
\begin{equation}\label{eq:Ansystem}
  f_i'+f_i \left( \sum_{j=1}^n f_{i+2j-1} - \sum_{j=1}^n f_{i+2j}
  \right)=\a_i,\qquad i=0,\dots,2n \mod (2n+1)
\end{equation}
subject to the normalization conditions
\begin{equation}
  \label{eq:alpha1}
f_0+\dots+f_{2n}= z,\qquad \a_0+\cdots + \a_{2n}=1.
\end{equation}

The symmetry group of this higher order system is the affine Weyl group $A_{2n}^{(1)}$, acting by \bk\ transformations as in \eqref{eq:BT}. The system passes the  Painlev\'e-Kowalevskaya test, \cite{veselov1993dressing}, and it is believed to possess the \p\ property.  It is thus considered a proper higher order generalization of $\rm{s}\Pfour$ \eqref{eq:P4system}, which corresponds to the special case $n=1$.

The next higher order system belonging to this hierarchy is the $A_4$-\p\ system, that has been studied by Filipuk and Clarkson  \cite{filipuk2008symmetric}, who provide several classes of rational solutions via an explicit Wronskian representation, and by Matsuda \cite{matsuda2012rational}, who uses the classical approach to identify the set of parameters that lead to rational solutions. However, a complete classification and explicit description of the rational solutions of $A_{2n}$-\p\ for $n\geq2$ is, to the best of our knowledge, still not available in the literature.

Of particular interest are the special polynomials associated with these rational solutions, whose zeros and poles structure shows extremely regular patterns in the complex plane, and have received a considerable amount of study, as mentioned above. 
We will show that all these special  polynomial  are only particular cases of a larger one.

Our approach for describing rational solutions to the Noumi-Yamada system makes no use of symmetry groups of \bk\ transformations, vertex operators, Hirota bilinear equations, etc. Instead, we will adopt the approach of Darboux dressing chains introduced by the Russian school \cite{adler1994nonlinear,veselov1993dressing}, which has received comparatively less attention in connection to \p\ systems, and the recent advances in the theory of exceptional polynomials \cite{gomez2009extended,gomez2010extension,gomez2016durfee}. 
Having said this, it would be a very interesting development to establish a dictionary between the symmetry group approach and the one presented in this paper.

The paper is organized as follows: in Section~\ref{sec:dressing} we introduce the equations for a dressing chain of Darboux transformations of Schr\"{o}dinger operators and prove that they are equivalent to the $A_{2n}$-\p\ system. These results are well known \cite{adler1994nonlinear} but recalling them is useful to fix notation and make the paper self contained.
In Section~\ref{sec:Maya} we explore the class of dressing chains built on rational extensions of the harmonic oscillator, which can be indexed by Maya diagrams. We introduce the key notion of \textit{cyclic Maya diagrams} and reformulate the problem of classifying rational solutions of the $A_{2n}$-\p\ system as that of classifying $(2n+1)$-cyclic Maya diagrams. In Section~\ref{sec:Mayacycles} we introduce the notion of genus and interlacing for Maya diagrams which allows us to achieve a complete classification of $p$-cyclic Maya diagrams for any period $p$. In Section\ref{sec:A4} , we focus on the $A_4$-\p\ system to give the class of rational solutions using the representation developed in the previous sections. Finally, we show some plots of the roots of these special solutions in the complex plane.

The purpose of this paper is to illustrate a new construction method and an explicit representation of rational solutions to higher order \p\ systems. We conjecture that this construction includes all possible rational solutions to the system. A proof of this fact requires new arguments than the ones developed in this paper, and remains for now an open question.


Even cyclic dressing chains provide higher order extensions of $\Pfive$. The situation for this even cyclic case corresponding to the $A_{2n+1}$-\p\ systems is considerably harder. Construction methods similar to the ones described here are available, but the class of dressing chains is larger, and it includes rational extensions of both the harmonic and the isotonic oscillator \cite{grandati2011solvable,grandati2012multistep}, thus described by universal characters (or pairs of Maya diagrams). Some rational solutions have been given by Tsuda \cite{tsuda2005universal} but the full classification for this case is still an open question. 


\section{Darboux dressing chains}\label{sec:dressing}

The theory of dressing chains, or sequences of Schr\"{o}dinger operators
connected by Darboux transformations was developed by Adler
\cite{adler1994nonlinear}, and Veselov and Shabat
\cite{veselov1993dressing}. The connection between dressing chains and
\p\ equations was already shown in \cite{adler1994nonlinear}
and it has been exploited by some authors 
\cite{takasaki2003spectral,tsuda2005universal,bermudez2012complexb,marquette2013one,marquette2013two,marquette2016,
sen2005darboux,WilloxHietarinta,MateoNegro}. 
This section follows mostly the early works of Adler, Veselov and Shabat.

Consider the following sequence of Schr\"{o}dinger operators
\begin{equation}\label{eq:Lseq}
  \cL_i = -D_z^2 + U_i ,\qquad D_z= \frac{{\rm d}}{{\rm d}z},\quad U_i=U_i(z),\quad
  i\in \Z
\end{equation}
where each operator is related to the next by a Darboux transformation, i.e. by the following factorization
\begin{equation}
  \label{eq:Dxform}
  \begin{aligned}
    \cL_i &= (D_z + w_i)(-D_z + w_i)+\lambda_i, \quad w_i = w_i(z),\\
    \cL_{i+1} &= (-D_z + w_i)(D_z + w_i)+\lambda_i.
  \end{aligned}
\end{equation}
It follows that the functions $w_i$ satisfy the Riccati equations
\begin{equation}\label{eq:Riccati}
 w_i' + w_i^2  = U_i - \lambda_i,\quad -w_i' +w_i^2 = U_{i+1}- \lambda_i.
 \end{equation}
Equivalently, $w_i$ are the log-derivatives of $\psi_i$, the seed function of the Darboux transformation that maps $\cL_i$ to $\cL_{i+1}$
\begin{equation}
  \label{eq:Lipsii}
  \cL_i\psi_i = \lambda_i\psi_i,\qquad\text{where } w_i = \frac{\psi_i'}{\psi_i}.
\end{equation}
Using \eqref{eq:Lseq} and \eqref{eq:Dxform}, the potentials of the dressing chain are related by
\begin{align}
  U_{i+1} &= U_i - 2 w'_i, \label{eq:Uplus1}\\
  U_{i+n} &=U_i - 2 \left( w'_i+ \cdots + w'_{i+n-1}\right),\quad
  n\geq 2 \label{eq:Uplusn}
\end{align}
If we eliminate the potentials in \eqref{eq:Riccati} and set
\begin{equation}
  \label{eq:alphaidef}
  a_i =  \lambda_{i} - \lambda_{i+1}
\end{equation}
the following chain of coupled equations is obtained
\[
(w_i + w_{i+1})' + w_{i+1}^2 - w_i^2 = a_i ,\quad i\in \Z
\]
Before continuing, note that this infinite chain of equations has the
evident reversal symmetry
\begin{equation}
  \label{eq:reversal1}
  w_i \mapsto -w_{-i},\qquad a_i \mapsto -a_{-i}.
\end{equation}

This infinite chain of equations closes and becomes a finite
dimensional system of \odes\ if a cyclic condition is imposed on the
potentials of the chain
\begin{equation}\label{eq:shift}
  U_{i+p} = U_i+\Delta,\quad i\in \Z
 \end{equation}
for some $p\in \N$ and $\Delta \in \C$.  If this holds, then
necessarily $w_{i+p}=w_i$, $a_{i+p}=a_i$, and
\begin{equation}
  \label{eq:Deltasumalpha}
 \Delta= -(a_0 + \cdots + a_{p-1}). 
\end{equation}

\begin{definition}\label{def:wchain}
  A $p$-cyclic Darboux dressing chain (or factorization chain) with
  shift $\Delta$ is a sequence of $p$ functions $w_0,\ldots, w_{p-1}$
  and complex numbers $a_0,\ldots, a_{p-1}$ that satisfy
  the following coupled system of $p$ Riccati-like \odes\
\begin{equation}
  \label{eq:wfchain}
  (w_i + w_{i+1})' + w_{i+1}^2 - w_i^2 = a_i ,\qquad
  i=0,1,\ldots, p-1 \mod(p)  
\end{equation}
subject to the condition \eqref{eq:Deltasumalpha}.
\end{definition}
Note that
transformation 
\begin{equation}
  \label{eq:reversal2}
  w_i \mapsto -w_{-i},\quad a_i \mapsto -a_{-i},\quad
  \Delta\mapsto -\Delta
\end{equation}
projects the reversal symmetry to the finite-dimensional system
\eqref{eq:wfchain}.  Moreover, for $j=0,1\ldots, p-1$ we also have the
cyclic symmetry
\[ w_i \mapsto w_{i + j},\quad a_i \mapsto a_{i+j},\quad
\Delta \mapsto \Delta \qquad i=0,\ldots p-1 \mod(p) \]
In the classification of solutions to \eqref{eq:wfchain} it will be
convenient to regard two solutions related by a reversal symmetry or
by a cyclic permutation as being equivalent.

Adding the $p$ equations \eqref{eq:wfchain} we immediately obtain a
first integral of the system
\[ \sum_{j=0}^{p-1} w_j= \tfrac12z\sum_{j=0}^{p-1} a_j=
-\tfrac12{\Delta}z.\]

\begin{remark}
We assume throughout this paper that $\Delta\neq0$, which is the only case for which the dressing chain in Definition~\ref{def:wchain} leads to solutions of the $A_N$-Painlev\'e system with normalization \eqref{eq:alpha1} (and thus to higher order generalizations of $\Pfour$). In the $\Delta=0$ case, the dressing chain \eqref{eq:wfchain} defines a completely integrable system whose general solution can be expressed in terms of elliptic and theta functions, \cite{veselov1993dressing}, the corresponding Schr\"odinger operators belonging to the class of finite-gap potentials. The  rational solutions in this class can be expressed in terms of Burchnall-Chaundy \cite{burchnall1930set} (or Adler-Moser \cite{adler1978class})  polynomials, which can be obtained by confluent Darboux-Crum transformations of the free potential at zero energy, \cite{refDG86}.
\end{remark}
In the $\Delta\neq0$ case,  the $A_{2n}$-\p\ system \eqref{eq:Ansystem} and the cyclic dressing chain \eqref{eq:wfchain} are related by the following proposition.

\begin{prop}\label{prop:wtof}
If the tuple of functions and complex numbers $(w_0,\dots,w_{2n}|a_0,\dots,a_{2n})$  satisfies a $(2n+1)$-cyclic Darboux dressing chain with shift $\Delta\neq 0$ as per Definition~\ref{def:wchain}, then the tuple $\left(f_0,\dots,f_{2n}\,\big|\,\a_0,\dots, \a_{2n}\right)$
with
\begin{eqnarray}
f_i(z)&=& c \,(w_i + w_{i+1})\left(cz\right),\qquad i=0,\dots,2n\mod(2n+1),\\
 \a_i&=&c^2 a_i,\\
 c^2&=&-\frac{1}{\Delta}
\end{eqnarray}
 solves the $A_{2n}$-\p\ system \eqref{eq:Ansystem} with normalization \eqref{eq:alpha1}.
\end{prop}
\begin{proof}
The linear transformation
\begin{equation}\label{eq:wtof}
f_i=w_i+w_{i+1}, \qquad i=0,\dots,2n\mod(2n+1)
\end{equation}
is invertible (only in the odd case $p=2n+1$), the inverse transformation being
\begin{equation}\label{eq:ftow}
  w_i= \tfrac{1}{2} \sum_{j=0}^{2n} (-1)^j f_{i+j}, \qquad
  i=0,\dots,2n\mod(2n+1) 
\end{equation}
They imply the relations
\begin{equation}\label{eq:wdif}
  w_{i+1}-w_i = \sum_{j=0}^{2n-1} (-1)^j f_{i+j+1}, \qquad
  i=0,\dots,2n\mod(2n+1). 
\end{equation}
Inserting \eqref{eq:wtof} and \eqref{eq:wdif} into the equations of the cyclic dressing chain \eqref{eq:wfchain} leads to the $A_{2n}$-\p\ system \eqref{eq:Ansystem}. For any constant $c\in\mathbb C$, the scaling transformation
\[f_i\mapsto c f_i,\quad z\mapsto cz,\quad \a_i\mapsto c^2 \a_i \]
preserves the form of the equations \eqref{eq:Ansystem}. The choice $c^2=-\frac{1}{\Delta}$ ensures that the normalization \eqref{eq:alpha1} always holds, for dressing chains with different shifts $\Delta$.
\end{proof}

\begin{remark}
$(2n)$-cyclic dressing chains and $A_{2n-1}$-\p\ systems are also related, but the mapping is given by a rational rather than a linear function. A full treatment of this even cyclic case (which includes  $\Pfive$ and its higher order hierarchy) is considerably harder and shall be treated elsewhere.
\end{remark}

The problem now becomes that of finding and classifying cyclic dressing chains, i.e. Schr\"{o}dinger operators and sequences of Darboux transformations that reproduce the initial potential up to an additive shift $\Delta$ after a fixed given number of transformations.

The theory of exceptional polynomials is intimately related with families of Schr\"{o}dinger operators connected by Darboux transformations \cite{gomez2013conjecture,garcia2016bochner}. Constructing cyclic dressing chains on this class of potentials becomes a feasible task, and knowledge of the effect of rational Darboux transformations on the potentials suggests that the only family of potentials to be considered in the case of odd cyclic dressing chains are the rational extensions of the harmonic oscillator \cite{gomez2013rational}, which are exactly solvable potentials whose eigenfunctions are expressible in terms of exceptional Hermite polynomials.

Each potential in this class can be indexed by a finite set of integers (specifying the sequence of Darboux transformations applied on the harmonic oscillator that lead to the potential), or equivalently by a Maya diagram, which becomes a  very useful representation to capture a notion of equivalence and relations of the type \eqref{eq:shift}.

As mentioned before, the fact that all rational odd cyclic dressing chains (and equivalently rational solutions to the $A_{2n}$-\p\ system) must \textit{necessarily} belong to this class remains an open question. We conjecture that this is indeed the case, and no  rational solutions other than the ones described in the following sections exist.

\section{Cyclic Maya diagrams and rational extensions of the Harmonic oscillator}\label{sec:Maya}

In this Section we construct odd cyclic dressing chains on potentials belonging to the class of rational extensions of the harmonic oscillator. Every such potential is represented by a Maya diagram, a rational Darboux transformation acting on this class will be a flip operation on a Maya diagram and cyclic Darboux chains correspond to cyclic Maya diagrams. With this representation, the main problem of constructing rational cyclic Darboux chains becomes purely algebraic and combinatorial.

Following Noumi \cite{noumi2004painleve}, we define a Maya diagram in the following manner.
\begin{definition}
  A Maya diagram is a set of integers $M\subset\Z$ that contains a
  finite number of positive integers, and excludes a finite number of
  negative integers.  We will use $\cM$ to denote the set of all Maya
  diagrams.
\end{definition}

\begin{definition}\label{def:index}
Let $m_1>m_2>\cdots$ be the elements of a Maya diagram $M$ arranged in decreasing order. By assumption, there exists a unique integer $s_M\in \Z$ such that $m_i = -i+s_M$ for all
$i$ sufficiently large. We define $s_M$ to be the index of $M$.
\end{definition}

We visualize a Maya diagram as a horizontally extended sequence of
$\boxdot$ and $\emptybox$ symbols with the filled symbol $\boxdot$ in
position $i$ indicating membership $i\in M$. The defining assumption
now manifests as the condition that a Maya diagram begins with an
infinite filled $\boxdot$ segment and terminates with an infinite
empty $\emptybox$ segment.

\begin{definition}

Let $M$ be a Maya diagram, and 
\[ M_-= \{ -m-1 \colon m\notin M, m<0\},\qquad M_+ = \{ m\colon m\in
M\,, m\geq 0 \}. \]
Let $s_1>s_2>\cdots > s_p$ and $t_1> t_2>\dots> t_q$ be the
elements of $M_-$ and $M_+$ arranged in descending order. 
 
We define the \textit{Frobenius symbol} of $M$ to be the double
list $(s_1,\ldots, s_p \mid t_q,\ldots, t_1)$.
\end{definition}
It is not hard to show that $s_M=q-p$ is the index of $M$.  The classical Frobenius symbol
\cite{andrews2004integer,olsson1994combinatorics,andrews1998theory} corresponds to the zero index case where $q=p$.
If $M$ is a Maya diagram, then for any $k\in \Z$ so is
\[ M+k = \{ m+k \colon m\in M \}.\]
The behaviour of the index $s_M$ under translation of $k$ is given by
\begin{equation}\label{eq:indexshift}
M'=M+k\quad \Rightarrow \quad s_{M'}=s_M+k.
\end{equation}
We will refer to an equivalence class of Maya diagrams related by such
shifts as an \textit{unlabelled Maya diagram}. One can visualize the
passage from an unlabelled to a labelled Maya diagram as the choice of
placement of the origin.

A Maya diagram $M\subset \Z$ is said to be in standard form if $p=0$
and $t_q>0$.  Visually, a Maya diagram in standard form has only
filled boxes $\boxdot$ to the left of the origin and one empty box
$\emptybox$ just to the right of the origin. Every unlabelled Maya
diagram permits a unique placement of the origin so as to obtain a
Maya diagram in standard form.

In \cite{gomez2016durfee} it was shown that to every Maya diagram we
can associate a polynomial called a Hermite pseudo-Wronskian.
\begin{definition}
  Let $M$ be a Maya diagram and $(s_1,\dots,s_r|t_q,\dots,t_1)$ its
  corresponding Frobenius symbol. Define the polynomial
  \begin{equation}\label{eq:pWdef1} \H_M(z) = \ex{-rz^2}\Wr[ \ex{z^2}
    \th_{s_1},\ldots, \ex{z^2} \th_{s_r}, \h_{t_q},\ldots \h_{t_1} ],
  \end{equation} where $\Wr$ denotes the Wronskian determinant of the
  indicated functions, and
  \begin{equation}
    \label{eq:thndef}
    \th_n(z)={\rm i}^{-n} \h_{n}({\rm i}z)
  \end{equation}
  is the $n\supth$ degree conjugate Hermite polynomial.
\end{definition}

The polynomial nature of $\H_M(z)$ becomes evident once we
represent it using a slightly different determinant.

\begin{prop} The Wronskian $\H_M(z)$ admits the following alternative
  determinantal representation
  \begin{equation}\label{eq:pWdef2} \H_M(z) =
    \begin{vmatrix} \th_{s_1} & \th_{s_1+1} & \ldots &
\th_{s_1+r+q-1}\\ \vdots & \vdots & \ddots & \vdots\\ \th_{s_r} &
\th_{s_r+1} & \ldots & \th_{s_r+r+q-1}\\ \h_{t_q} & D_{z} \h_{t_q} &
\ldots & D_{z}^{r+q-1}\h_{t_q}\\ \vdots & \vdots & \ddots & \vdots\\
\h_{t_1} & D_{z} \h_{t_1} & \ldots & D_{z}^{r+q-1}\h_{t_1}
    \end{vmatrix}
  \end{equation}

\end{prop} 
\noindent
The term Hermite
pseudo-Wronskian was coined in \cite{gomez2016durfee} because  \eqref{eq:pWdef2} is a mix of a Casoratian and a Wronskian
determinant.  For all Maya diagrams in the same equivalence class, their associated Hermite pseudo-Wronskians enjoy a very simple relation.

\begin{prop}[\cite{gomez2016durfee} Theorem 1]\label{prop:equiv}
  Let  $\widehat{\H}_M(z)$ be the normalized pseudo-Wronskian
  \begin{equation}
    \label{eq:hHdef}
    \widehat{\H}_M(z) = \frac{(-1)^{rq}\H_M(z)}{\prod_{1\leq i<j\leq r} (2s_j-2s_i)\prod_{1\leq
        i<j\leq q}
      (2 t_i-2t_j)}.
  \end{equation}
Then for any Maya diagram $M$ and $k\in\Z$ we have
  \begin{equation} \label{eq:HMequiv}
       \widehat{\H}_M(z) =  \widehat{\H}_{M+k}(z).
  \end{equation}
\end{prop}
The remarkable aspect of equation \eqref{eq:HMequiv} is that the identity involves determinants of different sizes. Note that the statement of this proposition is slightly different than the original result proved in \cite{gomez2016durfee}, due to the introduction of \textit{normalized} pseudo-Wronskians \eqref{eq:hHdef} to achieve strict equality in \eqref{eq:HMequiv} rather than just equality up to a multiplicative constant.
As mentioned above, every unlabelled Maya diagram contains a Maya diagram in standard form, and its associated Hermite pseudo-Wronskian \eqref{eq:pWdef1} is just an ordinary Wronskian determinant whose entries are Hermite polynomials. This will not be in general the smallest determinant in the equivalence class. The procedure to find the smallest equivalent determinant was given in \cite{gomez2016durfee}. 

Due to Proposition \ref{prop:equiv}, we could restrict the analysis without loss of generality to Maya diagrams in standard form and Wronskians of Hermite polynomials, but we will employ the general notation as it 
brings conceptual clarity to the description of cyclic Maya diagrams.

We will now introduce and study a class of potentials for
Schr\"{o}dinger operators that will be used as building blocks for
cyclic dressing chains.  Rational extensions of classic potentials
have been studied in a number of papers \cite{grandati2012multistep,grandati2011solvable, odake2013extensions,marquette2013two,bagchi2015rational} and their set of eigenfunctions are expressible in terms of exceptional orthogonal polynomials.

\begin{definition}
A rational extension of the harmonic
oscillator is a potential of the form
\[ U(z) = z^2 + \frac{a(z)}{b(z)},\] 
with  $a(x)$, $b(x)$  polynomials with $\deg a\leq \deg b$,
that is \textit{exactly solvable by polynomials}.  This means that for
all but finitely many $n\in\mathbb N$, the operator $\cL=-D^2+U(z)$ has
formal eigenfunctions of the form
\[ \psi_n = \mu(z) y_n(z),\]
where $\mu(z)$ is a fixed function and where $y_n(z)$ are polynomials
of degree $n$.
\end{definition}

If $b(z)$ has no real zeros, then $\cL$ is a Sturm-Liouville operator on $\R$ with quasi-polynomial eigenfunctions.  Exact solvability by polynomials is a very stringent property, which is equivalent to trivial monodromy \cite{refDG86,oblomkov1999monodromy}. In fact, the next Proposition proved in\cite{gomez2013rational} states that rational extensions of the harmonic oscillator can be put in one to one correspondence with Maya diagrams.

\begin{prop}[\cite{gomez2013rational} Theorem 1.1]
  \label{prop:ratext}
Let $M\subset \Z$ be a Maya diagram.  Define
\begin{equation}
  \label{eq:UMdef}
  U_M(z) = z^2 - 2 D_z^2 \log \H_M(z) + 2s_M,
\end{equation}
where $\H_M(z)$ is the corresponding pseudo-Wronskian \eqref{eq:pWdef1}--\eqref{eq:pWdef2}, and $s_M\in \Z$ is the index of $M$.
Up to an additive constant, every rational extension of the harmonic
  oscillator takes the form \eqref{eq:UMdef}.
\end{prop}

The class of Schr\"{o}dinger operators with potentials that are rational extensions of the harmonic oscillator is invariant
under a certain class of  rational Darboux transformations, which we
now describe.

\begin{definition}
We define the flip at position $m\in \Z$ to be the involution
$\phi_m:\cM\to \cM$ defined by
\begin{equation}\label{eq:flipdef}
 \phi_m : M \mapsto
\begin{cases}
   M \cup \{ m \}, & \text{if}\quad m\notin M, \\
   M \setminus \{ m \},\quad & \text{if}\quad m\in M.
\end{cases}\qquad M\in \cM.
\end{equation}
\end{definition}
\noindent
In the first case, we say that $\phi_m$ acts on $M$ by a
state-deleting transformation ($\emptybox\to \boxdot$).  In the second
case, we say that $\phi_m$ acts by a state-adding transformation
($\boxdot\to\emptybox$).

It can be shown that every quasi-rational eigenfunction
\cite{gomez2004supersymmetry,gomez2004darboux} of $ \cL=-D_z^2+ U_M(z)$
has the form
\begin{equation}
  \label{eq:seedfunc}
  \psi_{M,m} = \exp(\tfrac12\ep z^2)\frac{\H_{\phi_m(M)}(z)}{\H_M(z)}, \qquad m\in \Z,
\end{equation}
with 
\[ \ep = \begin{cases}
  -1, & \text{if}\quad m\notin M, \\
  +1,\quad & \text{if}\quad m\in M,
\end{cases} 
\]
Explicitly, we have
\begin{equation}
    \label{eq:Mneigenfunc}
    \cL \psi_{M,m}  = (2m+1) \psi_{M,m} ,\quad m\in \Z.    
  \end{equation}

\begin{remark}
The seed eigenfunctions \eqref{eq:seedfunc} include the true eigenfunctions of $\cL$ plus another set of formal non square-integrable eigenfunctions, sometimes known in the physics literature as \textit{virtual states},\cite{odake2013krein,odake2011exactly}. For a correct spectral theoretic interpretation one needs to ensure that the potential $U_M$ is regular, i.e. that $\H_M(z)$ has no zeros in $\R$. The set of Maya diagrams for which $\H_M(z)$ has no real zeros was characterized (in a more general setting) independently by Krein \cite{krein1957continuous} and Adler \cite{adler1994modification}, while the number of real zeros for $\H_M$ was given in \cite{garcia2015oscillation}. However, for the purpose of this paper it is convenient to stay within a purely formal setting and keep the whole class of potentials $U_M$, regardless of whether they have real poles or not.
\end{remark}

The relation between dressing chains of Darboux transformations for the class of operators \eqref{eq:UMdef} and flip operations on Maya diagrams is made explicit by the following proposition. 

\begin{prop}
  \label{prop:UMflip}
  Two Maya diagrams $M, M'$ are related by a flip \eqref{eq:flipdef}
  if and only if their associated rational extensions $U_M,U_{M'}$, see \eqref{eq:UMdef}, are
  connected by a Darboux transformation \eqref{eq:Uplus1}.
\end{prop}
\begin{proof}
  Suppose that $m\notin M$ and that $M' = M \cup \{ m\}$ is a
  state-deleting flip transformation of $M$. The seed function for the factorization is $\psi_{M,m}$ defined in \eqref{eq:seedfunc}.   Set
  \begin{equation}
    \label{eq:f1statedelete}
    w_{M,m}(z) = \frac{\psi_{M,m}'(z)}{\psi_{M,m}(z)} = -z + \frac{\H_{M'}'(z)}{\H_{M'}(z)} - \frac{\H_{M}'(z)}{\H_{M}(z)}.
  \end{equation}
  Since
  \[ s_{M'} = s_M +1,\]
  by \eqref{eq:UMdef}, we have 
  \begin{equation}
    \label{eq:UMM'}
    \tfrac12\big[U_{M'}(z)-U_{M}(z)\big] = 1+ \ddz \left(\frac{\H_{M}'(z)}{\H_{M}(z)}-  \frac{\H_{M'}'(z)}{\H_{M'}(z)} \right) = -w_{M,m}'(z),
  \end{equation}
so that \eqref{eq:Uplus1} holds.
  Conversely, suppose that $M$ and $M'$ are such that \eqref{eq:UMM'}
  holds for some $w = w(z)$. If we define
  \[ w(z) = \frac{\psi'(z)}{\psi(z)},\qquad
   \psi(z) = \ex{-\tfrac12z^2}\frac{ \H_{M'}(z)}{\H_{M}(z)} ,\]
   then $\psi $ must be a quasi-rational seed function for $U_M$ and
   it follows by \eqref{eq:seedfunc} of Proposition \ref{prop:ratext} 
   that $M'= M\cup \{m\}$ for some $m\notin M$. The corresponding result for state-adding Darboux transformations is done in a similar way.
\end{proof}

We see thus that the class of rational extensions of the harmonic
oscillator is indexed by Maya diagrams, and that the Darboux
transformations that preserve this class can be described by flip
operations on Maya diagrams.
It now becomes feasible to characterize cyclic dressing chains built
on this class of potentials.

\begin{definition}
  \label{def:multiflip}
  For $p\in \N$ let $\cZ_p$ denote the set of all subsets of $\Z$
  having cardinality $p$.
  For $\bmu= \{ \mu_1,\ldots, \mu_p\}\in \cZ_p$ we now define
  $\phi_{\bmu}$ to be the multi-flip
   \begin{equation}
     \label{eq:phimudef}
     \phi_{\bmu}= \phi_{\mu_1} \circ \cdots \circ \phi_{\mu_{p}}.
\end{equation}
\end{definition}

We are now ready to introduce the basic concept of this section.

\begin{definition}
  We say that $M$ is $p$-cyclic with shift $k$, or $(p,k)$ cyclic, if
  there exists a $\bmu \in \cZ_p$ such that
  \begin{equation}
    \label{eq:cyclicMdef}
    \phi_\bmu(M) = M+k.
  \end{equation}
  We will say that $M$ is $p$-cyclic if it is $(p,k)$ cyclic for some
  $k\in \Z$.
\end{definition}

\begin{prop}
  \label{prop:muM1M2}
  For Maya diagrams $M,M'\in \cM$, we define the set $ \Upsilon(M,M') $ as the symmetric difference between $M$ and $M'$:
  \begin{equation}
    \label{eq:muM1M2}
    \Upsilon(M,M') = (M \setminus M') \cup (M'\setminus M).
  \end{equation}
  Then the multi-flip $\phi_\bmu$ where $\bmu=\Upsilon(M,M')$ is the unique
  multi-flip such that $ M' = \phi_{\bmu}(M)$ and 
  $M = \phi_{\bmu}(M')$.
\end{prop}
As an immediate corollary, we have the following.
\begin{prop}
  Let $k$ be a non-zero integer.  Every Maya diagram $M\in \cM$ is
  $(p,k)$ cyclic where $p$ is the cardinality of
  $\bmu=\Upsilon(M,M+k)$.
\end{prop}


We are now able to establish the link between cyclic Maya diagrams and cyclic
dressing chains composed of rational extensions of the harmonic
oscillator.
\begin{prop}
  \label{prop:Mwcorrespondence}
  Let $M\in \cM$ be a Maya diagram, $k$ a non-zero integer, and $p$
  the cardinality of $\bmu=\Upsilon(M,M+k)$. Let
  $\bmu = \{\mu_0,\ldots, \mu_{p-1}\}$ be an arbitrary enumeration of
  $\bmu$ and set
  \begin{equation}
    \label{eq:Mchain}
    M_0 = M,\quad M_{i+1} = \phi_{\mu_i}(M_{i}),\qquad
    i=0,1,\ldots, p-1
  \end{equation}
  so that $M_p = M_0+k$ by construction.  Set
  \begin{align}
    &w_i(z)= s_{i} z+ \frac{\H_{M_{i+1}}'(z)}{\H_{M_{i+1}}(z)}- \frac{\H_{M_{i}}'(z)}{\H_{M_{i}}(z)},\qquad i=0,\dots,p-1. \label{eq:HM2w}\\
    &a_i=2(\mu_i-\mu_{i+1}), \label{eq:mu2alpha}
      \intertext{where}
     \label{eq:sign}
    &s_i=\begin{cases}
      -1, & \textit{if}\quad  \mu_i\notin M, \\
    +1, & \textit{if}\quad \mu_i\in M,
  \end{cases}
\end{align}
and 
\[ \mu_p = \mu_0 + k.\]
Then, $(w_0,\ldots, w_{p-1}; a_0,\ldots, a_{p-1})$
constitutes a rational solution to the $p$-cyclic dressing chain
\eqref{eq:wfchain} with shift $\Delta=2k$.
\end{prop}

\begin{proof}
The result follows from the structure of the seed eigenfunctions \eqref{eq:seedfunc} with eigenvalues given by \eqref{eq:Mneigenfunc}, after applying \eqref{eq:Lipsii} and \eqref{eq:alphaidef}.
The sign of $s_{i}$ indicates whether the $(i+1)$-th step of the chain that takes $\cL_i$ to $\cL_{i+1}$ is a state-adding $(+1)$ or state-deleting $(-1)$ transformation.
\end{proof}

The remaining part of the construction is
to classify cyclic Maya diagrams for any given (odd) period, which we
tackle next.  Under the correspondence described by Proposition
\ref{prop:Mwcorrespondence}, the reversal symmetry
\eqref{eq:reversal2} manifests as the transformation
\[ (M_0,\ldots, M_p) \mapsto (M_p,\ldots, M_0),\quad (\mu_1,\ldots,
\mu_{p}) \mapsto (\mu_{p},\ldots, \mu_1),\quad k\mapsto -k.\]
In light of the above remark, there is no loss of generality if we
restrict our attention to cyclic Maya diagrams with a positive shift
$k>0$.

\section{Classification of cyclic Maya diagrams}\label{sec:Mayacycles}

In this section we introduce the key concepts of \textit{genus} and
\textit{interlacing} to achieve a full classification of cyclic Maya
diagrams.



For $\bbeta\in \cZ_{2g+1}$ define the Maya diagram
\begin{equation}
  \label{eq:MBi}
  \Xi(\bbeta)= (-\infty,\beta_0) \cup [\beta_1,\beta_2) \cup
  \ \cdots \cup [\beta_{2g-1},\beta_{2g})
\end{equation}
where
\[ [m,n) = \{ j\in \Z \colon m\leq j < n\}\]
and where $\beta_0<\beta_1<\cdots < \beta_{2g}$ is the strictly increasing
enumeration of $\bbeta$.

\begin{prop}
  Every Maya diagram $M\in \cM$ has a unique representation of the
  form $M=\Xi(\bbeta)$ where $\bbeta$ is a set of integers of odd
  cardinality $2g+1$.
\end{prop}
\begin{definition}\label{def:genus}
  We call the integer $g\geq 0$ the genus of $M= \Xi(\bbeta)$ and
  $(\beta_0,\beta_1,\ldots, \beta_{2g})$ the block coordinates of $M$.
\end{definition}
\noindent

\begin{prop}
  \label{prop:1cyclic}
  Let $M = \Xi(\bbeta)$ be a Maya diagram specified by its block
  coordinates.   We then have
  \[ \bbeta = \Upsilon(M,M+1).\]
\end{prop}
\begin{proof}
  Observe that
  \[ M+1 = (-\infty, \beta_0] \cup (\beta_1,\beta_2] \cup \cdots \cup
  (\beta_{2g-1}, \beta_{2g}],\]
  where
  \[ (m,n] = \{ j\in \Z \colon m<j\leq n \}.\]
  It follows that
  \begin{align*}
    (M+1)\setminus M &= \{ \beta_0, \ldots, \beta_{2g} \}\\
    M\setminus (M+1) &= \{ \beta_1,\ldots, \beta_{2g-1} \}.
  \end{align*}
  The desired conclusion follows immediately.
\end{proof}


Let $\cM_g$ denote the set of Maya diagrams of genus $g$.  The above
discussion may be summarized by saying that the mapping \eqref{eq:MBi}
defines a bijection $\Xi:\cZ_{2g+1} \to \cM_g$, and that the block
coordinates are precisely the flip sites required for a translation
$M\mapsto M+1$.

\begin{remark}
  To motivate Definition \ref{def:genus}, it is perhaps more illustrative to
  understand the visual meaning of the genus of $M$, see
  Figure~\ref{fig:genusM}.  After removal of the initial infinite
  $\boxdot$ segment and the trailing infinite $\emptybox$ segment, a
  Maya diagram consists of alternating empty $\emptybox$ and filled
  $\boxdot$ segments of variable length.  The genus $g$ counts the
  number of such pairs.  
  The even block coordinates $\beta_{2i}$ indicate the starting
  positions of the empty segments, and the odd block coordinates
  $\beta_{2i+1}$ indicated the starting positions of the filled
  segments.  Also, note that $M$ is in standard form if and only if
  $\beta_0=0$.
\end{remark}

\begin{figure}[h]
\begin{tikzpicture}[scale=0.6]

\draw  (1,1) grid +(15 ,1);

\path [fill] (0.5,1.5) node {\huge ...} 
++(1,0) circle (5pt) ++(1,0) circle (5pt)  ++(1,0) circle (5pt) 
++(1,0) circle (5pt) ++(1,0) circle (5pt) 
++(2,0) circle (5pt) ++(1,0) circle (5pt) 
++ (3,0) circle (5pt)  ++(1,0) circle (5pt)   ++ (1,0) circle (5pt) 
++ (3,0) node {\huge ...} +(1,0) node[anchor=west] { $M =  (-\infty,\beta_0)\cup [ \beta_1,\beta_2) \cup [ \beta_3,\beta_4)$};

\draw[line width=1pt] (4,1) -- ++ (0,1.5);

\foreach \x  in {-3,...,11} 	\draw (\x+4.5,2.5)  node {$\x$};
\path (6.5,0.5) node {$\beta_0$} ++ (1,0) node {$\beta_1$}
++ (2,0) node {$\beta_2$}++ (2,0) node {$\beta_3$}++ (3,0) node {$\beta_4$}
;
\end{tikzpicture}

\caption{Block coordinates $(\beta_0,\ldots, \beta_4) = (2,3,5,7,10)$
  of a genus $2$ Maya diagrams.  Note that the genus is both the
  number of finite-size empty blocks and the number of finite-size
  filled blocks.}\label{fig:genusM}
\end{figure}
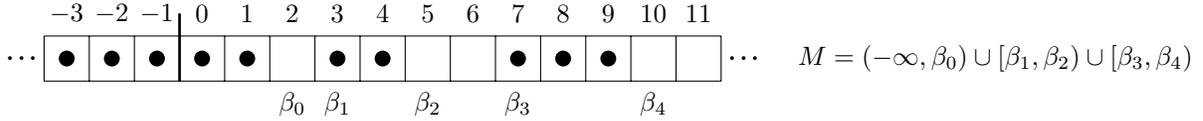

The next concept we need to introduce is the interlacing and modular
decomposition. 
\begin{definition}\label{def:interlacing}
  Fix a $k\in \N$ and let $M^{(0)}, M^{(1)},\ldots M^{(k-1)}\subset \Z$ be sets
  of integers.  We define the interlacing of these to be the set
  \begin{equation}\label{eq:interlacing} \Theta\left(M^{(0)}, M^{(1)},\ldots M^{(k-1)}\right)
    = \bigcup_{i=0}^{k-1} (k M^{(i)} +i),
 \end{equation}
 where
 \[ kM +j = \{ km + j \colon m\in M \},\quad M\subset \Z.\]
 Dually, given a set of integers $M\subset \Z$ and a $k\in \N$ define
 the sets
 \[ M^{(i)} = \{ m\in \Z \colon km+i \in M\},\quad i=0,1,\ldots, k-1.\]
 We will call the $k$-tuple of sets $\left(M^{(0)}, M^{(1)},\ldots M^{(k-1)}\right)$ the
 $k$-modular decomposition of $M$.
\end{definition}


The following result follows directly from the above definitions.
\begin{prop}
  We have $M=\Theta\left(M^{(0)}, M^{(1)},\ldots M^{(k-1)}\right)$ if and only if
  $\left(M^{(0)}, M^{(1)},\ldots M^{(k-1)}\right)$ is the $k$-modular decomposition of $M$.
\end{prop}

Even though the above operations of interlacing and modular
decomposition apply to general sets, they have a well defined
restriction to Maya diagrams.  Indeed, it is not hard to check that if
$M=\Theta\left(M^{(0)}, M^{(1)},\ldots M^{(k-1)}\right)$ and $M$ is a Maya diagram, then
 $M^{(0)}, M^{(1)},\ldots M^{(k-1)}$ are also Maya diagrams.  Conversely, if the latter are all Maya
diagrams, then so is $M$.  Another important case concerns the
interlacing of finite sets.  The definition \eqref{eq:interlacing}
implies directly that if $\bmu^{(i)} \in \cZ_{p_i},\; i=0,1,\ldots, k-1$
then
\[ \bmu = \Theta\left(\bmu^{(0)}, \ldots , \bmu^{(k-1)}\right) \]
is a finite set of cardinality $p=p_0+\cdots + p_{k-1}$.

Visually, each of the $k$ Maya diagrams is dilated by a factor of $k$,
shifted by one unit with respect to the previous one and superimposed,
so the interlaced Maya diagram incorporates the information from
$M^{(0)}, \ldots M^{(k-1)}$ in $k$ different modular classes. An example can
be seen in Figure~\ref{fig:interlacing}. In other words, the
interlaced Maya diagram is built by copying sequentially a filled or
empty box as determined by each of the $k$ Maya diagrams.

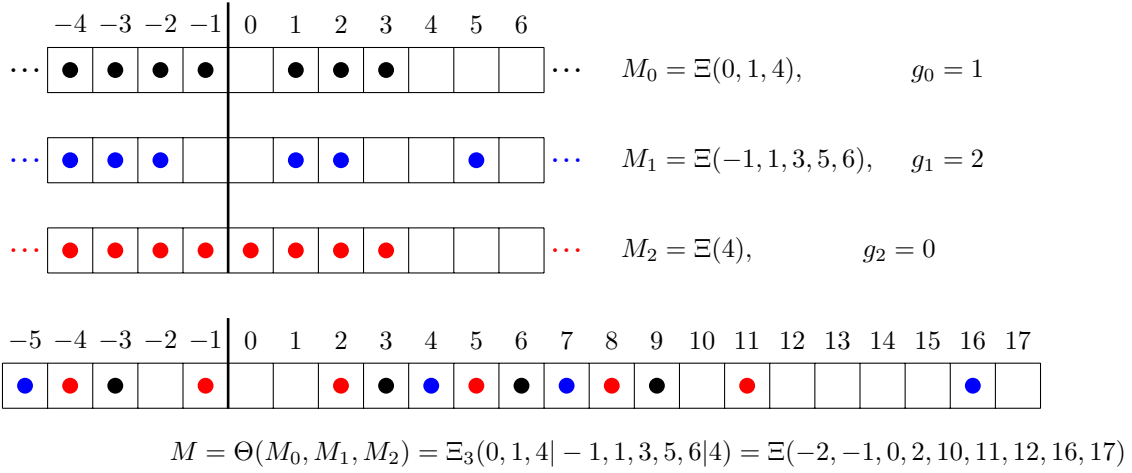
\begin{figure}[h]
\begin{tikzpicture}[scale=0.6]

\draw  (1,3) grid +(11 ,1);

\path [fill,color=black] (0.5,3.5) node {\huge ...} 
++(1,0) circle (5pt) ++(1,0) circle (5pt)  ++(1,0) circle (5pt) 
++(1,0) circle (5pt)
++(2,0) circle (5pt) ++(1,0) circle (5pt)  ++(1,0) circle (5pt) 
++ (4,0) node {\huge ...} +(1,0) node[anchor=west,color=black] { $M_0 = \Xi(0,1,4),\;\qquad\quad\,\,\,\, g_0 = 1$}; 

\draw[line width=1pt] (5,3) -- ++ (0,2);

\foreach \x in {-4,...,6} 	\draw (\x+5.5,4.5)  node {$\x$};

\draw  (1,1) grid +(11 ,1);

\path [fill,color=blue] (0.5,1.5) node {\huge ...} 
++(1,0) circle (5pt) ++(1,0) circle (5pt)  ++(1,0) circle (5pt) 
++(3,0) circle (5pt) ++(1,0) circle (5pt) 
++ (3,0) circle (5pt)  ++(2,0) node {\huge ...} +(1,0) node[anchor=west,color=black] { $M_1 = \Xi(-1,1,3,5,6),\;\quad g_1 = 2$}; 

\draw[line width=1pt] (5,1) -- ++ (0,2);

\draw  (1,-1) grid +(11 ,1);

\path [fill,color=red] (0.5,-0.5) node {\huge ...} 
++(1,0) circle (5pt) ++(1,0) circle (5pt)  ++(1,0) circle (5pt) 
++(1,0) circle (5pt) ++(1,0) circle (5pt)  ++(1,0) circle (5pt) 
++(1,0) circle (5pt)  ++(1,0) circle (5pt) 
++(4,0) node {\huge ...} +(1,0) node[anchor=west,color=black] { $M_2 = \Xi(4),\qquad\qquad g_2 = 0$}; 

\draw[line width=1pt] (5,-1) -- ++ (0,2);

\draw  (0,-4) grid +(23 ,1);
\foreach \x in {-5,...,17} 	\draw (\x+5.5,2.5-5)  node {$\x$};
\draw[line width=1pt] (5,-4) -- ++ (0,2);

\path [fill,color=black] (2.5,-3.5)   
circle (5pt) ++(6,0) circle (5pt)  
++(3,0) circle (5pt) ++(3,0) circle (5pt) ;

\path [fill,color=blue] 
(0.5,-3.5) circle (5pt) ++(9,0) circle (5pt)  
++(3,0) circle (5pt) ++(9,0) circle (5pt) ;

\path [fill,color=red] 
(1.5,-3.5) circle (5pt) ++(3,0) circle (5pt)  
++(3,0) circle (5pt) ++(3,0) circle (5pt)
++(3,0) circle (5pt) ++(3,0) circle (5pt) ;

\draw (3.5,-5) node[right] {$M= \Theta(M_0,M_1,M_2)=\Xi_3(0,1,4|-1,1,3,5,6|4) = \Xi(-2,-1,0,2,10,11,12,16,17)$}; 
\end{tikzpicture}

\caption{Interlacing of three Maya diagrams with genus $1,2$ and $0$ with block coordinates and \mbox{$3$-block} coordinates for the interlaced Maya diagram.}\label{fig:interlacing}
\end{figure}

\begin{remark}
Modular decomposition of Maya diagrams has been considered previously by Noumi in his book  \cite{noumi2004painleve} (see Proposition 7.12), although in a different context: that of studying the effect of B\"acklund transformations on the Maya diagrams. In the present context of dressing chains and rational solutions, Tsuda \cite{tsuda2005universal} has also employed the notation in \eqref{eq:interlacing} for the interlacing of $N$ genus-0 Maya diagrams, which correspond to $N$-reduced partitions. This particular family of rational solutions correspond to the signature class $(1,1,\dots,1)$ with the highest shift $k=2n+1$ (see Section~\ref{sec:A4} below), i.e. the generalization of Okamoto polynomials.
\end{remark}

Equipped with these notions of genus and interlacing, we are now ready
to state the main result for the classification of cyclic Maya
diagrams.


\begin{thm}
  \label{thm:Mp}
  Let $M=\Theta\left(M^{(0)}, M^{(1)},\ldots M^{(k-1)}\right)$ be the $k$-modular decomposition of a given Maya diagram $M$.  Let $g_i$ be the genus
  of $M^{(i)},\; i=0,1,\ldots, k-1$.  Then, $M$ is $(p,k)$-cyclic where
  \begin{equation}
    \label{eq:pgi}
    p = p_0+p_1+\cdots + p_{k-1},\qquad p_i = 2g_i + 1.
  \end{equation}
\end{thm}
\begin{proof}
  Let $\bbeta^{(i)}=\Upsilon\left(M^{(i)},M^{(i)}+1\right) \in \cZ_{p_i}$ be the block
  coordinates of $M^{(i)},\; i=0,1,\ldots, k-1$.  Consider the interlacing
  $\bmu = \Theta\left(\bbeta^{(0)}, \ldots , \bbeta^{(k-1)}\right)$.  From
  Proposition \ref{prop:1cyclic} we have that,
  \[ \phi_{\bbeta^{(i)}}\left( M^{(i)}\right)= M^{(i)} +1.\]
  so it follows that
  \begin{align*}
    \phi_{\bmu}(M)&= \phi_{\Theta\left(\bbeta^{(0)}, \ldots , \bbeta^{(k-1)}\right)}\Theta\left( M^{(0)} , \ldots ,
                  M^{(k-1)}\right) \\   
    &= \Theta\left( \phi_{\bbeta^{(0)}}(M^{(0)}) , \ldots ,
                    \phi_{\bbeta^{(k-1)}}(M^{(k-1)})\right) \\
                  &= \Theta\left(M^{(0)}+1 , \ldots , M^{(k-1)} + 1\right) \\
                  &= \Theta\left(M^{(0)}, \ldots , M^{(k-1)}\right) + k \\
                  &= M+k.
  \end{align*}
  Therefore, $M$ is $(p,k)$ cyclic where the value of $p$ agrees with
  \eqref{eq:pgi}.
\end{proof}

Theorem~\ref{thm:Mp} sets the way to classify cyclic Maya diagrams for
any given period $p$.
\begin{cor}\label{cor:k}
  For a fixed period $p\in \N$, there exist $p$-cyclic Maya diagrams with
  shifts $k=p,p-2,\dots,\lfloor p/2\rfloor$, and no other positive shifts are  possible.
\end{cor}

\begin{remark}
  The highest shift $k=p$ corresponds to the interlacing of $p$ trivial (genus
  0) Maya diagrams.
\end{remark}

We now introduce a combinatorial system for describing rational
solutions of $p$-cyclic factorization chains.  First, we require a
suitably generalized notion of block coordinates suitable for
describing $p$-cyclic Maya diagrams.
\begin{definition}
  Let $M=\Theta\left(M^{(0)},\ldots M^{(k-1)}\right)$ be a $k$-modular decomposition of a
  $(p,k)$ cyclic Maya diagram.  For $i=0,1,\ldots, k-1$ let
  $\bbeta^{(i)}= \left(\beta^{(i)}_{0}, \ldots,
    \beta^{(i)}_{p_i-1}\right)$
  be the block coordinates of $M^{(i)}$ enumerated in increasing order.
  In light of the fact that
  \[ M = \Theta\left(\Xi(\bbeta^{(0)}), \ldots , \Xi(\bbeta^{(k-1)})\right),\]
  we will refer to the concatenated sequence
  \begin{align*}
    (\beta_0,\beta_1,\ldots, \beta_{p-1}) 
    &= (\bbeta^{(0)} | \bbeta^{(1)} | \ldots |
      \bbeta^{(k-1)}) \\
    &=  \left( \beta^{(0)}_{0}, \ldots, \beta^{(0)}_{p_0-1} |
      \beta^{(1)}_{0}, \ldots, \beta^{(1)}_{p_1-1} 
      | \ldots | \beta^{(k-1)}_{0},\ldots, \beta^{(k-1)}_{p_{k-1}-1}\right) 
  \end{align*}
  as the $k$-block coordinates of $M$.  Formally, the correspondence
  between $k$-block coordinates and Maya diagram is described by the
  mapping
  \[ \Xi_k\colon \cZ_{2g_0+1} \times \cdots \times \cZ_{2g_{k-1}+1} \to \cM \]
  with action
  \[ \Xi_k \colon (\bbeta^{(0)} | \bbeta^{(1)} | \ldots |
      \bbeta^{(k-1)})\mapsto \Theta\left(\Xi(\bbeta^{(0)}), \ldots ,
    \Xi(\bbeta^{(k-1)})\right) \]
\end{definition}

\begin{definition}
  Fix a $k\in \N$. For $m\in \Z$ let $[m]_k\in \{ 0,1,\ldots, k-1 \}$
  denote the residue class of $m$ modulo division by $k$.  For
  $m,n \in \Z$ say that $m \preccurlyeq_k n$ if and only if
  \[ [m]_k < [n]_k,\quad \text{ or } \quad [m]_k = [n]_k\; \text{ and
  } m\leq n.\]
  In this way, the transitive, reflexive relation $\preccurlyeq_k$
  forms a total order on $\Z$.
\end{definition}

\begin{prop}
  Let $M$ be a $(p,k)$ cyclic Maya diagram.  There exists a unique
  $p$-tuple of integers $(\mu_0,\ldots, \mu_{p-1})$ strictly ordered
  relative to $\preccurlyeq_k$ such that
  \begin{equation}
    \label{eq:bmuMk}
    \phi_\bmu(M) = M+k 
  \end{equation}
\end{prop}
\begin{proof}
  Let
  $(\beta_0,\ldots, \beta_{p-1}) =  (\bbeta^{(0)} | \bbeta^{(1)} | \ldots |   \bbeta^{(k-1)})$ be the $k$-block coordinates of $M$.
  Set \[ \bmu = \Theta\left(\bbeta^{(0)}, \ldots , \bbeta^{(k-1)}\right) \]
  so that \eqref{eq:bmuMk} holds by the proof to Theorem \ref{thm:Mp}.
  The desired enumeration of $\bmu$ is given by
  \[ (k \beta_0,\ldots, k \beta_{p-1}) + (0^{p_0}, 1^{p_1}, \ldots,
  (k-1)^{p_{k-1}}) \]
  where the exponents indicate repetition.  Explicitly,
  $(\mu_0,\ldots, \mu_{p-1})$ is given by
  \[ \left(k\beta^{(0)}_{0},\ldots, k\beta^{(0)}_{p_0-1}, k\beta^{(1)}_{0} + 1 ,\ldots,
  k\beta^{(1)}_{p_1-1} +1 , \ldots , k \beta^{(k-1)}_{0} + k-1, \ldots,
  k\beta^{(k-1)}_{p_{k-1}-1} + k-1 \right) .\]
\end{proof}
\begin{definition}
  In light of \eqref{eq:bmuMk} we will refer to the just defined tuple
  $(\mu_0,\mu_1,\ldots, \mu_{p-1})$ as the $k$-canonical flip sequence
  of $M$ and refer to the tuple $(p_0,p_1,\ldots, p_{k-1})$ as the
  $k$-signature of $M$.
\end{definition}

By Proposition \ref{prop:Mwcorrespondence} a rational solution of the
$p$-cyclic dressing chain requires a $(p,k)$ cyclic Maya diagram, and
an additional item data, namely a fixed ordering of the canonical
flip sequence.
We will specify such ordering  as
\[ \bmu_\bpi = (\mu_{\pi_0},\ldots, \mu_{\pi_{p-1}}) \]
where $\bpi=(\pi_0,\ldots, \pi_{p-1})$ is a permutation of
$(0,1,\ldots, p-1)$. With this notation, the chain of Maya diagrams
described in Proposition \ref{prop:Mwcorrespondence}
is generated as
\begin{equation}\label{eq:picycle}
 M_0 = M,\qquad M_{i+1} =  \phi_{\mu_{\pi_i}}(M_i),\qquad i=0,1,\ldots, p-1.
 \end{equation}
\begin{remark}\label{rem:normalization}
Using a translation it is possible to normalize $M$ so that
$\mu_{0} = 0$.  Using a cyclic permutation and it is possible to
normalize $\bpi$ so that $\pi_p=0$.  The net effect of these two
normalizations is to ensure  that
$M_0,M_1,\ldots, M_{p-1}$ have standard form.
\end{remark}

\begin{remark}
In the discussion so far we have imposed the hypothesis that the
sequence of flips that produces a translation $M\mapsto M+k$ does not
contain any repetitions.  However, in order to obtain a full
classification of rational solutions, it will be necessary to account
for degenerate chains which include multiple flips at the same site.

To that end it is necessary to modify Definition \ref{def:multiflip}
to allow $\bmu$ to be a multi-set\footnote{A multi-set is
  generalization of the concept of a set that allows for multiple
  instances for each of its elements.}, and to allow
$\mu_0,\mu_1,\ldots, \mu_{p-1}$ in \eqref{eq:MBi} to be merely a
non-decreasing sequence.  
This has the effect of permitting $\emptybox$ and $\boxdot$ segments of
zero length wherever $\mu_{i+1} = \mu_i$.  The $\Xi$-image of such a
non-decreasing sequence is not necessarily a Maya diagram of genus
$g$, but rather a Maya diagram whose genus is bounded above by $g$.

It is no longer possible to assert that there is a unique $\bmu$ such
that $\phi_\bmu(M) = M+k$, because it is possible to augment the
non-degenerate $\bmu = \Upsilon(M,M+k)$ with an arbitrary number of
pairs of flips at the same site to arrive at a degenerate $\bmu'$ such
that $\phi_{\bmu'}(M) = M+k$ also.  The rest of the theory remains
unchanged.
\end{remark}

\section{Rational solutions of $A_4$-\p}\label{sec:A4}

In this section we will put together all the results derived above in order to describe an effective way of labelling and constructing all the rational solutions to the $A_{2k}$-\p\ system based on cyclic dressing chains of rational extensions of the harmonic oscillator. We conjecture that the construction described below covers all rational solutions to such systems. As an illustrative example, we describe all rational solutions to the $A_4$-\p\ system, and we furnish examples in each signature class.

For odd $p$, in order to specify a Maya $p$-cycle, or equivalently a rational solution of a $p$-cyclic dressing chain, we need to specify three items of data:
\begin{enumerate}[(i)]
\item a signature sequence $(p_0,\ldots, p_{k-1})$ consisting of odd positive integers that sum to $p$. This sequence determines the genus of the $k$ interlaced Maya diagrams that give rise to a $(p,k)$-cyclic Maya diagram $M$. The possible values of $k$ are given by Corollary~\ref{cor:k}.
\item Once the signature is fixed, we need to specify the $k$-block coordinates \[(\beta_0,\dots,\beta_{p-1})=(\bbeta^{(0)}|\ldots| \bbeta^{(k-1)})\] where $\bbeta^{(i)}=(\beta^{(i)}_0,\dots,\beta^{(i)}_{p_i})$ are the block coordinates that define each of the interlaced Maya diagrams $M^{(i)}$. These two items of data specify uniquely a $(p,k)$-cyclic Maya diagram $M$, and a canonical flip sequence $\bmu=(\beta_0,\dots,\beta_{p-1})$ . The next item specifies a given $p$-cycle that contains $M$.
\item Once the $k$-block coordinates and canonical flip sequence $\bmu$ are fixed, we still have the freedom to choose a permutation  $\bpi\in S_p$ of $(0,1,\ldots, p-1)$ that specifies the actual flip sequence $\bmu_\bpi$, i.e. the order in which the flips in the canonical flip sequence are applied to build the Maya $p$-cycle.

\end{enumerate}

For any signature of a Maya $p$-cycle, we need to specify the $p$ integers in the canonical flip sequence, but following Remark~\ref{rem:normalization}, we can get rid of translation invariance by setting $\mu_0=\beta^{(0)}_0=0$, leaving only $p-1$ free integers. Moreover, we can restrict ourselves to permutations such that $\bpi_p=0$ in order to remove the invariance under cyclic permutations. The remaining number of degrees of freedom is $p-1$, which (perhaps not surprisingly) coincides with the number of generators of the symmetry group $A^{(1)}_{p-1}$. This is a strong indication that the class described above captures a generic orbit of a seed solution under the action of the symmetry group.


We now illustrate the general theory by describing the rational
solutions of the $A^{(1)}_4$- \p\ system, whose equations are given by

\begin{align}\label{eq:A4system}
f_0' + f_0(f_1-f_2+f_3-f_4) &= \a_0, \nonumber\\
f_1' + f_1(f_2-f_3+f_4-f_0) &= \a_1,\nonumber \\
f_2' + f_2(f_3-f_4+f_0-f_1) &= \a_2, \\
f_3' + f_3(f_4-f_0+f_1-f_2) &= \a_3, \nonumber \\
f_4' + f_4(f_0-f_1+f_2-f_3) &= \a_4,\nonumber 
\end{align}
with normalization conditions
\[ f_0+f_1+f_2+f_3+f_4=z,\qquad \a_0+\a_1+\a_2+\a_3+\a_4=1.\]
This system has the ``seed solutions"
\begin{align*}
(f_0,f_1,f_2,f_3,f_4)&=(z,0,0,0,0), && (\a_0,\a_1,\a_2,\a_3,\a_4)=(1,0,0,0,0),\\
(f_0,f_1,f_2,f_3,f_4)&=(\tfrac13z,\tfrac13z,\tfrac13z,0,0), && (\a_0,\a_1,\a_2,\a_3,\a_4)=(\tfrac13,\tfrac13,\tfrac13,0,0),\\
(f_0,f_1,f_2,f_3,f_4)&=(\tfrac15z,\tfrac15z,\tfrac15z,\tfrac15z,\tfrac15z), && (\a_0,\a_1,\a_2,\a_3,\a_4)=(\tfrac15,\tfrac15,\tfrac15,\tfrac15,\tfrac15),
\end{align*}
and permutations thereof.
\begin{thm}\label{thm:A4}
  Rational solutions of the $A^{(1)}_4$-\p\ system \eqref{eq:A4system} correspond to chains of $5$-cyclic
  Maya diagrams belonging to one of the following signature classes:
  \[ (5), (3,1,1), (1,3,1), (1,1,3), (1,1,1,1,1).\]
  With the normalization $\bpi_4=0$ and $\mu_0=0$, each rational
  solution may be uniquely labelled by one of the above signatures,
  a 4-tuple of arbitrary non-negative integers $(n_1, n_2, n_3, n_4)$, and  a
  permutation $(\pi_0, \pi_1,\pi_2, \pi_3)$ of $(1,2,3,4)$.  For
  each of the above signatures, the corresponding $k$-block coordinates
  of the initial $5$-cyclic Maya diagram are then given by
  \begin{align*}
&k=1 & &(5)& &(0,n_1,n_1+n_2, n_1+n_2+n_3, n_1+n_2+n_3+n_4) \\
&k=3 & &(3,1,1) && (0, n_1 , n_1+ n_2 | n_3 | n_4)\\
&k=3 & &(1,3,1) & &(0| n_1, n_1+n_2, n_1+n_2+n_3 | n_4)\\
&k=3 & &(1,1,3) & &(0| n_1 | n_2, n_2+n_3, n_2+n_3+n_4 )\\
&k=5 & &(1,1,1,1,1) && (0| n_1 | n_2| n_3| n_4)\\
  \end{align*}
\end{thm}

We show specific examples  with shifts $k=1,3$ and $5$ and signatures $(5)$, $(1,1,3)$ and $(1,1,1,1,1)$.

\begin{example}\label{ex:51}
We construct a $(5,1)$-cyclic Maya diagram in the signature class $(5)$ by choosing $(n_1,n_2,n_3,n_4)=(2,3,1,1)$, which means that the first Maya diagram in the cycle is $M_0=\Xi(0,2,5,6,7)$, depicted in the first row of Figure~\ref{fig:51cyclic}.  
The canonical flip sequence is  $\bmu=(0,2,5,6,7)$.  We choose the permutation
  $(34210)$, which gives the chain of Maya diagrams shown in Figure
  \ref{fig:51cyclic}.  Note that the permutation specifies the sequence of block coordinates that get shifted by one at each step of the cycle. This type of solutions with signature $(5)$ were already studied in \cite{filipuk2008symmetric}, and they are based on a genus 2 generalization of the generalized Hermite polynomials that appear in the solution of \PIV ($A_2$-\p).
\begin{figure}[ht]
  \centering
\begin{tikzpicture}[scale=0.5]

  \path [fill] (0.5,6.5) ++(3,0)
  circle (5pt) ++(1,0)
  circle (5pt) ++ (1,0) circle (5pt)++ (2,0)
  circle (5pt) ++ (3,0)  node[anchor=west] {
    $M_0=\Xi\,(0,2,5,6,7)$};

  \path [fill] (0.5,5.5) ++(3,0) circle (5pt) ++(1,0) circle (5pt) ++
  (1,0) circle (5pt) ++ (5,0)  node[anchor=west] {
    $M_1=\Xi\,(0,2,5,7,7)$};

  \path [fill] (0.5,4.5) ++(3,0)
  circle (5pt) ++(1,0)
  circle (5pt) ++ (1,0) circle (5pt)++ (3,0)
  circle (5pt) ++ (2,0)  node[anchor=west] {
    $M_2=\Xi\,(0,2,5,7,8)$};

  \path [fill] (0.5,3.5) ++(3,0)
  circle (5pt) ++(1,0)
  circle (5pt) ++ (1,0) circle (5pt) ++(1,0) circle (5pt) ++ (2,0)
  circle (5pt) ++ (2,0)  node[anchor=west] {
    $M_3=\Xi\,(0,2,6,7,8)$};

  \path [fill] (0.5,2.5)  ++(4,0)
  circle (5pt) ++ (1,0) circle (5pt) ++(1,0) circle (5pt) ++ (2,0)
  circle (5pt) ++ (2,0)  node[anchor=west] {
    $M_4=\Xi\,(0,3,6,7,8)$};

  \path [fill] (0.5,1.5)  ++(1,0) circle (5pt) ++(3,0)
  circle (5pt) ++ (1,0) circle (5pt) ++(1,0) circle (5pt) ++ (2,0)
  circle (5pt) ++ (2,0)  node[anchor=west] {
    $M_5=\Xi\,(1,3,6,7,8)=M_0+1$};

  \path [fill] (0.5,1.5) ++(0,0) circle (5pt)
  ++(0,1) circle (5pt)
  ++(0,1) circle (5pt)
  ++(0,1) circle (5pt)
  ++(0,1) circle (5pt)
  ++(0,1) circle (5pt);

  \draw  (0,1) grid +(10 ,6);
  \draw[line width=2pt] (1,1) -- ++ (0,6);
  \draw[line width=2pt] (9,1) -- ++ (0,6);

  \foreach \x in {-1,...,8} \draw (\x+1.5,0.5)  node {$\x$};

\end{tikzpicture}
  
\caption{A Maya $5$-cycle with shift $k=1$ for the choice $(n_1,n_2,n_3,n_4)=(2,3,1,1)$ and permutation $\bpi=(34210)$. }
  \label{fig:51cyclic}
\end{figure}
\end{example}

We shall now provide the explicit construction of the rational solution to the $A_4$-\p\ system \eqref{eq:A4system}, by using Proposition~\ref{prop:Mwcorrespondence} and Proposition~\ref{prop:wtof}. The permutation $\bpi=(34210)$ on the canonical sequence $\bmu=(0,2,5,6,7)$ produces the flip sequence $\bmu_\bpi=(6,7,5,2,0)$, so that the values of the $a_i$ parameters given by \eqref{eq:mu2alpha} become 
$(a_0,a_1,a_2,a_3,a_4)=(-2,4,6,4,-14)$.
The pseudo-Wronskians corresponding to each Maya diagram in the cycle are ordinary Wronskians, which will always be the case with the normalization imposed in Remark~\ref{rem:normalization}. They read (see Figure~\ref{fig:51cyclic}):
\begin{align*}
\H_{M_0}(z)&=\Wr(H_2,H_3,H_4,H_6),\\
\H_{M_1}(z)&=\Wr(H_2,H_3,H_4),\\
\H_{M_2}(z)&=\Wr(H_2,H_3,H_4,H_7),\\
\H_{M_3}(z)&=\Wr(H_2,H_3,H_4,H_5,H_7),\\
\H_{M_4}(z)&=\Wr(H_3,H_4,H_5,H_7),\\
\end{align*}
where $H_n=H_n(z)$ is the $n\supth$ Hermite polynomial. Following Proposition~\ref{prop:Mwcorrespondence}, the rational solution to the dressing chain is given by the tuple $(w_0,w_1,w_2,w_3,w_4|a_0,a_1,a_2,a_3,a_4)$, where $a_i$ and $w_i$ are given by \eqref{eq:HM2w}--\eqref{eq:mu2alpha} as:
\begin{align*}
w_0(z)&=z+\ddz\Big[\log \H_{M_1}(z) - \log \H_{M_0}(z)\Big],&& a_0=-2,\\
w_1(z)&=-z+ \ddz\Big[\log \H_{M_2}(z)-\log \H_{M_1}(z)\Big],&& a_1=4,\\
w_2(z)&=-z+\ddz\Big[\log \H_{M_3}(z)- \log \H_{M_2}(z)\Big],&& a_2=6,\\
w_3(z)&=z+ \ddz\Big[\log \H_{M_4}(z) - \log \H_{M_3}(z)\Big],&& a_3=4,\\
w_4(z)&=-z+ \ddz\Big[\log \H_{M_0}(z) - \log \H_{M_4}(z)\Big],&& a_4=-14.
\end{align*}
Finally, Proposition~\ref{prop:wtof} implies that the corresponding rational solution to the $A_4$-\p\ system \eqref{eq:Ansystem} is given by the tuple $(f_0,f_1,f_2,f_3,f_4|\a_0,\a_1,\a_2,\a_3,\a_4)$, where
\begin{align*}
f_0(z)&= \ddz\Big[\log \H_{M_2}(c_1 z)- \log \H_{M_0}(c_1 z)\Big],&& \alpha_0=1,\\
f_1(z)&=z+\ddz\Big[\log \H_{M_3}(\cc{1}z) -\log \H_{M_1}(\cc{1}z)\Big],&& \alpha_1=-2,\\
f_2(z)&= \ddz\Big[\log \H_{M_4}(\cc{1}z) - \log \H_{M_2}(\cc{1}z)\Big],&& \alpha_2=-3,\\
f_3(z)&=\ddz\Big[\log \H_{M_0}(\cc{1}z) - \log \H_{M_3}(\cc{1}z)\Big],&& \alpha_3=-2,\\
f_4(z)&=\ddz\Big[\log \H_{M_1}(\cc{1}z) - \log \H_{M_4}(\cc{1}z)\Big],&& \alpha_4=7.
\end{align*}
with $\cc{1}^2=-\frac{1}{2}$.

\begin{example}
 We construct a degenerate example belonging to the $(5)$ signature class, by choosing 
 $(n_1,n_2,n_3,n_4)=(1,1,2,0)$. The presence of $n_4=0$ means that the first Maya diagram has genus 1 instead of the generic genus 2, with block coordinates given by $M_0=\Xi\,(0,1,2,4,4)$. The canonical flip sequence  $\bmu=(0,1,2,4,4)$ contains two flips at the same site, so it is not unique. Choosing the permutation $(42130)$ produces the chain of Maya  diagrams shown in Figure \ref{fig:51cyclicdegen}.  The explicit construction of the rational solutions follows the same steps as in the previous example, and we shall omit it here. It is worth noting, however, that due to the degenerate character of the chain, three linear combinations of $f_0,\dots,f_4$ will provide a solution to the lower rank $A_2$-\p. If the two flips at the same site are performed consecutively in the cycle, the embedding of $A_2^{(1)}$ into $A_4^{(1)}$ is trivial and corresponds to setting two consecutive $f_i$ to zero. This is not the case in this example, as the flip sequence is $\bmu_\bpi=(4,2,1,4,0)$, which produces a non-trivial embedding.
\begin{figure}[ht]
  \centering
\begin{tikzpicture}[scale=0.5]

  \path [fill] (0.5,6.5) 
  ++(2,0)circle (5pt)
  ++(5,0)  node[anchor=west] {  $M_0=\Xi\,(0,1,2,4,4)$};

  \path [fill] (0.5,5.5) 
  ++(2,0) circle (5pt) 
  ++(3,0) circle (5pt)
  ++ (2,0)  node[anchor=west] {$M_1=\Xi\,(0,1,2,4,5)$};

  \path [fill] (0.5,4.5) 
  ++(2,0) circle (5pt)
  ++(1,0) circle (5pt)
  ++(2,0) circle (5pt)
  ++(2,0)  node[anchor=west] {  $M_2=\Xi\,(0,1,3,4,5)$};

  \path [fill] (0.5,3.5) 
  ++(3,0) circle (5pt)
  ++(2,0) circle (5pt)
  ++(2,0)  node[anchor=west] {  $M_3=\Xi\,(0,2,3,4,5)$};

  \path [fill] (0.5,2.5) 
  ++(3,0) circle (5pt)
  ++(4,0)  node[anchor=west] {  $M_4=\Xi\,(0,2,3,5,5)$};

  \path [fill] (0.5,1.5) 
  ++(1,0) circle (5pt)
  ++(2,0) circle (5pt)
  ++(4,0)  node[anchor=west] {  $M_5=\Xi\,(1,2,3,5,5)=M_0+1$};

  \path [fill] (0.5,1.5) ++(0,0) circle (5pt)
  ++(0,1) circle (5pt)
  ++(0,1) circle (5pt)
  ++(0,1) circle (5pt)
  ++(0,1) circle (5pt)
  ++(0,1) circle (5pt);

  \draw  (0,1) grid +(7 ,6);
  \draw[line width=2pt] (1,1) -- ++ (0,6);
  \draw[line width=2pt] (6,1) -- ++ (0,6);

  \foreach \x in {-1,...,5} \draw (\x+1.5,0.5)  node {$\x$};

\end{tikzpicture}
  
  \caption{A degenerate  Maya $5$-cycle with $k=1$ for the choice  $(n_1,n_2,n_3,n_4)=(1,1,2,0)$ and permutation $\bpi=(42130)$.}
  \label{fig:51cyclicdegen}
\end{figure}
\end{example}

\begin{example}
We construct a $(5,3)$-cyclic Maya diagram in the signature class $(1,1,3)$ by choosing $(n_1,n_2,n_3,n_4)=(3,1,1,2)$, which means that the first Maya diagram has $3$-block coordinates $(0|3|1,2,4)$. The canonical flip sequence is given by $\bmu=\Theta\,(0|3|1,2,4)=(0, {\color{red}10},{\color{blue} 5,8,14})$.
The permutation $(41230)$ gives the chain of Maya diagrams shown in Figure
  \ref{fig:53cyclic}. Note that, as in Example~\ref{ex:51}, the permutation specifies the order in which the $3$-block coordinates are shifted by +1 in the subsequent steps of the cycle. This type of solutions in the signature class $(1,1,3)$ were not given in \cite{filipuk2008symmetric}, and they are new to the best of our knowledge.

\begin{figure}[ht]
  \centering
\begin{tikzpicture}[scale=0.5]
  \path [fill,red] (0.5,6.5)  
  ++(2,0) circle (5pt)
  ++(3,0) circle (5pt)
  ++(3,0) circle (5pt) ;

  \path [fill,blue] (0.5,6.5)  
  ++(3,0) circle (5pt)
  ++(6,0) circle (5pt) 
  ++(3,0) circle (5pt) ;

  \path (17.5,6.5)  node[anchor=west] {$M_0=\Xi_3\,(0|3|1,2,4)$};

  \path [fill,red] (0.5,5.5)  
  ++(2,0) circle (5pt)
  ++(3,0) circle (5pt)
  ++(3,0) circle (5pt) ;

  \path [fill,blue] (0.5,5.5)  
  ++(3,0) circle (5pt)
  ++(6,0) circle (5pt) 
  ++(3,0) circle (5pt) 
  ++(3,0) circle (5pt) ;
  \path (17.5,5.5)  node[anchor=west] {$M_1=\Xi_3\,(0|3|1,2,5)$};

  \path [fill,red] (0.5,4.5)  
  ++(2,0) circle (5pt)
  ++(3,0) circle (5pt)
  ++(3,0) circle (5pt)
  ++(3,0) circle (5pt) ;

  \path [fill,blue] (0.5,4.5)  
  ++(3,0) circle (5pt)
  ++(6,0) circle (5pt) 
  ++(3,0) circle (5pt) 
  ++(3,0) circle (5pt) ;
  \path (17.5,4.5)  node[anchor=west] {$M_2=\Xi_3\,(0|4|1,2,5)$};

  \path [fill,red] (0.5,3.5)  
  ++(2,0) circle (5pt)
  ++(3,0) circle (5pt)
  ++(3,0) circle (5pt)
  ++(3,0) circle (5pt) ;

  \path [fill,blue] (0.5,3.5)  
  ++(3,0) circle (5pt)
  ++(3,0) circle (5pt)
  ++(3,0) circle (5pt) 
  ++(3,0) circle (5pt) 
  ++(3,0) circle (5pt) ;
  \path (17.5,3.5)  node[anchor=west] {$M_3=\Xi_3\,(0|4|2,2,5)$};

  \path [fill,red] (0.5,2.5)  
  ++(2,0) circle (5pt)
  ++(3,0) circle (5pt)
  ++(3,0) circle (5pt)
  ++(3,0) circle (5pt) ;
  \path [fill,blue] (0.5,2.5)  
  ++(3,0) circle (5pt)
  ++(3,0) circle (5pt)
  ++(6,0) circle (5pt) 
  ++(3,0) circle (5pt) ;
  \path (17.5,2.5)  node[anchor=west] {$M_4=\Xi_3\,(0|4|2,3,5)$};

  \path [fill,black] (0.5,1.5)  
  ++(1,0) circle (5pt);  
  \path [fill,red] (0.5,1.5)  
  ++(2,0) circle (5pt)
  ++(3,0) circle (5pt)
  ++(3,0) circle (5pt)
  ++(3,0) circle (5pt) ;
  \path [fill,blue] (0.5,1.5)  
  ++(3,0) circle (5pt)
  ++(3,0) circle (5pt)
  ++(6,0) circle (5pt) 
  ++(3,0) circle (5pt) ;
  \path (17.5,1.5)  node[anchor=west] {$M_5=\Xi_3\,(1|4|2,3,5)=M_0+3$};

  \path [fill,blue] (0.5,1.5) ++(0,0) circle (5pt)
  ++(0,1) circle (5pt)
  ++(0,1) circle (5pt)
  ++(0,1) circle (5pt)
  ++(0,1) circle (5pt)
  ++(0,1) circle (5pt);

  \path [fill,red] (-.5,1.5) ++(0,0) circle (5pt)
  ++(0,1) circle (5pt)
  ++(0,1) circle (5pt)
  ++(0,1) circle (5pt)
  ++(0,1) circle (5pt)
  ++(0,1) circle (5pt);

  \path [fill,black] (-1.5,1.5) ++(0,0) circle (5pt)
  ++(0,1) circle (5pt)
  ++(0,1) circle (5pt)
  ++(0,1) circle (5pt)
  ++(0,1) circle (5pt)
  ++(0,1) circle (5pt);

  \draw  (-2,1) grid +(19 ,6);
  \draw[line width=2pt] (1,1) -- ++ (0,6);
  \draw[line width=2pt] (16,1) -- ++ (0,6);

  \foreach \x in {0,...,15} \draw (\x+1.5,0.5)  node {$\x$};
\end{tikzpicture}
  \caption{A Maya $5$-cycle with shift $k=3$ for the choice $(n_1,n_2,n_3,n_4)=(3,1,1,2)$ and permutation $\bpi=(41230)$. }
  \label{fig:53cyclic}
\end{figure}

We proceed to build the explicit rational solution to the $A_4$-\p\ system \eqref{eq:A4system}. In this case, the permutation $\bpi=(41230)$ on the canonical sequence $\bmu=(0,10,5,8,14)$ produces the flip sequence $\bmu_\bpi=(14,10,5,8,0)$, so that the values of the $a_i$ parameters given by \eqref{eq:mu2alpha} become 
$(a_0,a_1,a_2,a_3,a_4)=(8,10,-6,16,-34)$.
The pseudo-Wronskians corresponding to each Maya diagram in the cycle are ordinary Wronskians, which will always be the case with the normalization imposed in Remark~\ref{rem:normalization}. They read (see Figure~\ref{fig:53cyclic}):
\begin{align*}
\H_{M_0}(z)&=\Wr(H_1,H_2,H_4,H_7,H_8,H_{11}),\\
\H_{M_1}(z)&=\Wr(H_1,H_2,H_4,H_7,H_8,H_{11},H_{14}),\\
\H_{M_2}(z)&=\Wr(H_1,H_2,H_4,H_7,H_8,H_{10},H_{11},H_{14}),\\
\H_{M_3}(z)&=\Wr(H_1,H_2,H_4,H_5,H_7,H_8,H_{10},H_{11},H_{14}),\\
\H_{M_4}(z)&=\Wr(H_1,H_2,H_4,H_5,H_7,H_{10},H_{11},H_{14}),\\
\end{align*}
where $H_n=H_n(z)$ is the $n$-th Hermite polynomial. The rational solution to the dressing chain is given by the tuple $(w_0,w_1,w_2,w_3,w_4|a_0,a_1,a_2,a_3,a_4)$, where $a_i$ and $w_i$ are given by \eqref{eq:HM2w}--\eqref{eq:mu2alpha} as:
\begin{align*}
w_0(z)&=-z+\ddz\Big[\log \H_{M_1}(z) - \log \H_{M_0}(z)\Big],&& a_0=8,\\
w_1(z)&=-z+ \ddz\Big[\log \H_{M_2}(z)-\log \H_{M_1}(z)\Big],&& a_1=10,\\
w_2(z)&=-z+ \ddz\Big[\log \H_{M_3}(z) - \log \H_{M_2}(z)\Big],&& a_2=-6,\\
w_3(z)&=z+ \ddz\Big[\log \H_{M_4}(z) - \log \H_{M_3}(z)\Big],&& a_3=16,\\
w_4(z)&=-z+ \ddz\Big[\log \H_{M_0}(z) - \log \H_{M_4}(z)\Big],&& a_4=-34.
\end{align*}
Finally, Proposition~\ref{prop:wtof} implies that the corresponding rational solution to the $A_4$-\p\ system \eqref{eq:Ansystem} is given by the tuple $(f_0,f_1,f_2,f_3,f_4|\a_0,\a_1,\a_2,\a_3,\a_4)$, where
\begin{align*}
f_0(z)&=\tfrac13z+\ddz\Big[\log \H_{M_2}(\cc{2}z) - \log \H_{M_0}(\cc{2}z)\Big],&& \alpha_0=-\tfrac43,\\
f_1(z)&=\tfrac13z+\ddz\Big[\log \H_{M_3}(\cc{2}z) -\log \H_{M_1}(\cc{2}z)\Big],&& \alpha_1=-\tfrac53,\\
f_2(z)&= \ddz\Big[\log \H_{M_4}(\cc{2}z) - \log \H_{M_2}(\cc{2}z)\Big],&& \alpha_2=1,\\
f_3(z)&=\ddz\Big[\log \H_{M_0}(\cc{2}z) -\log \H_{M_3}(\cc{2}z)\Big],&& \alpha_3=-\tfrac83,\\
f_4(z)&=\tfrac13z+\ddz\Big[\log \H_{M_1}(\cc{2}z) - \log \H_{M_4}(\cc{2}z)\Big],&& \alpha_4=\tfrac{17}{3}.
\end{align*}
with $\cc{2}^2=-\tfrac16$.
\end{example}

\begin{example}
We construct a $(5,5)$-cyclic Maya diagram in the signature class $(1,1,1,1,1)$ by choosing $(n_1,n_2,n_3,n_4)=(2,3,0,1)$, which means that the first Maya diagram has $5$-block coordinates $(0|2|3|0|1)$. The canonical flip sequence is given by  $\bmu=\Theta\,(0|2|3|0|1)=(0,{\color{red}11},{\color{blue}17},{\color{brown}3},{\color{green}
    9})$. The permutation $(32410)$ gives the chain of Maya diagrams shown in Figure \ref{fig:55cyclic}.  Note that, as it happens in the previous examples, the permutation specifies the order in which the $5$-block coordinates are shifted by +1 in the subsequent steps of the cycle. This type of solutions with signature $(1,1,1,1,1)$ were already studied in \cite{filipuk2008symmetric}, and they are based on a generalization of the Okamoto polynomials that appear in the solution of \PIV ($A_2$-\p).
    
\begin{figure}[ht]
  \centering
\begin{tikzpicture}[scale=0.5]
  \path [fill,red] (0.5,6.5)  
  ++(2,0) circle (5pt)
  ++(5,0) circle (5pt);
  \path [fill,blue] (0.5,6.5)  
  ++(3,0) circle (5pt)
  ++(5,0) circle (5pt)
  ++(5,0) circle (5pt) ;
  \path [fill,green] (0.5,6.5)  
  ++(5,0) circle (5pt);
  \path (20.5,6.5)  node[anchor=west] {$M_0=\Xi_5\,(0|2|3|0|1)$};

  \path [fill,red] (0.5,5.5)  
  ++(2,0) circle (5pt)
  ++(5,0) circle (5pt);
  \path [fill,blue] (0.5,5.5)  
  ++(3,0) circle (5pt)
  ++(5,0) circle (5pt)
  ++(5,0) circle (5pt) ;
  \path [fill,brown] (0.5,5.5)  
  ++(4,0) circle (5pt);
  \path [fill,green] (0.5,5.5)  
  ++(5,0) circle (5pt);
  \path (20.5,5.5)  node[anchor=west] {$M_1=\Xi_5\,(0|2|3|1|1)$};

  \path [fill,red] (0.5,4.5)  
  ++(2,0) circle (5pt)
  ++(5,0) circle (5pt);
  \path [fill,blue] (0.5,4.5)  
  ++(3,0) circle (5pt)
  ++(5,0) circle (5pt)
  ++(5,0) circle (5pt) 
  ++(5,0) circle (5pt) ;
  \path [fill,brown] (0.5,4.5)  
  ++(4,0) circle (5pt);
  \path [fill,green] (0.5,4.5)  
  ++(5,0) circle (5pt);
  \path (20.5,4.5)  node[anchor=west] {$M_2=\Xi_5\,(0|2|4|1|1)$};

  \path [fill,red] (0.5,3.5)  
  ++(2,0) circle (5pt)
  ++(5,0) circle (5pt);
  \path [fill,blue] (0.5,3.5)  
  ++(3,0) circle (5pt)
  ++(5,0) circle (5pt)
  ++(5,0) circle (5pt) 
  ++(5,0) circle (5pt) ;
  \path [fill,brown] (0.5,3.5)  
  ++(4,0) circle (5pt);
  \path [fill,green] (0.5,3.5)  
  ++(5,0) circle (5pt)
  ++(5,0) circle (5pt);
  \path (20.5,3.5)  node[anchor=west] {$M_3=\Xi_5\,(0|2|4|1|2)$};

  \path [fill,red] (0.5,2.5)  
  ++(2,0) circle (5pt)
  ++(5,0) circle (5pt)
  ++(5,0) circle (5pt);
  \path [fill,blue] (0.5,2.5)  
  ++(3,0) circle (5pt)
  ++(5,0) circle (5pt)
  ++(5,0) circle (5pt) 
  ++(5,0) circle (5pt) ;
  \path [fill,brown] (0.5,2.5)  
  ++(4,0) circle (5pt);
  \path [fill,green] (0.5,2.5)  
  ++(5,0) circle (5pt)
  ++(5,0) circle (5pt);
  \path (20.5,2.5)  node[anchor=west] {$M_4=\Xi_5\,(0|3|4|1|2)$};

  \path [fill,black] (0.5,1.5)  
  ++(1,0) circle (5pt);  
  \path [fill,red] (0.5,1.5)  
  ++(2,0) circle (5pt)
  ++(5,0) circle (5pt)
  ++(5,0) circle (5pt);
  \path [fill,blue] (0.5,1.5)  
  ++(3,0) circle (5pt)
  ++(5,0) circle (5pt)
  ++(5,0) circle (5pt) 
  ++(5,0) circle (5pt) ;
  \path [fill,brown] (0.5,1.5)  
  ++(4,0) circle (5pt);
  \path [fill,green] (0.5,1.5)  
  ++(5,0) circle (5pt)
  ++(5,0) circle (5pt);
  \path (20.5,1.5)  node[anchor=west] {$M_5=\Xi_5(1|3|4|1|2)=M_0+5$};

  \path [fill,green] (0.5,1.5) ++(0,0) circle (5pt)
  ++(0,1) circle (5pt)
  ++(0,1) circle (5pt)
  ++(0,1) circle (5pt)
  ++(0,1) circle (5pt)
  ++(0,1) circle (5pt);
  \path [fill,brown] (-0.5,1.5) ++(0,0) circle (5pt)
  ++(0,1) circle (5pt)
  ++(0,1) circle (5pt)
  ++(0,1) circle (5pt)
  ++(0,1) circle (5pt)
  ++(0,1) circle (5pt);
  \path [fill,blue] (-1.5,1.5) ++(0,0) circle (5pt)
  ++(0,1) circle (5pt)
  ++(0,1) circle (5pt)
  ++(0,1) circle (5pt)
  ++(0,1) circle (5pt)
  ++(0,1) circle (5pt);

  \path [fill,red] (-2.5,1.5) ++(0,0) circle (5pt)
  ++(0,1) circle (5pt)
  ++(0,1) circle (5pt)
  ++(0,1) circle (5pt)
  ++(0,1) circle (5pt)
  ++(0,1) circle (5pt);

  \path [fill,black] (-3.5,1.5) ++(0,0) circle (5pt)
  ++(0,1) circle (5pt)
  ++(0,1) circle (5pt)
  ++(0,1) circle (5pt)
  ++(0,1) circle (5pt)
  ++(0,1) circle (5pt);

  \draw  (-4,1) grid +(24 ,6);
  \draw[line width=2pt] (1,1) -- ++ (0,6);
  \draw[line width=2pt] (19,1) -- ++ (0,6);

  \foreach \x in {0,...,18} \draw (\x+1.5,0.5)  node {$\x$};
\end{tikzpicture}
  \caption{A Maya $5$-cycle with shift $k=5$ for the choice $(n_1,n_2,n_3,n_4)=(2,3,0,1)$ and permutation $\bpi=(32410)$.}
  \label{fig:55cyclic}
\end{figure}

We proceed to build the explicit rational solution to the $A_4$-\p\ system \eqref{eq:A4system}. In this case, the permutation $\bpi=(32410)$ on the canonical sequence $\bmu=(0,11,17,3,9)$ produces the flip sequence $\bmu_\bpi=(3,17,9,11,0)$, so that the values of the $a_i$ parameters given by \eqref{eq:mu2alpha} become 
$(a_0,a_1,a_2,a_3,a_4)=(-28,16,-4,22,-16)$.
The pseudo-Wronskians corresponding to each Maya diagram in the cycle are ordinary Wronskians, which will always be the case with the normalization imposed in Remark~\ref{rem:normalization}. They read:
\begin{align*}
\H_{M_0}(z)&=\Wr(H_1,H_2,H_4,H_6,H_7,H_{12}),\\
\H_{M_1}(z)&=\Wr(H_1,H_2,H_3,H_4,H_6,H_7,H_{12}),\\
\H_{M_2}(z)&=\Wr(H_1,H_2,H_3,H_4,H_6,H_7,H_{12},H_{17}),\\
\H_{M_3}(z)&=\Wr(H_1,H_2,H_3,H_4,H_6,H_7,H_9,H_{12},H_{17}),\\
\H_{M_4}(z)&=\Wr(H_1,H_2,H_3,H_4,H_6,H_7,H_9,H_{11},H_{12},H_{17}),\\
\end{align*}
where $H_n=H_n(z)$ is the $n$-th Hermite polynomial. The rational solution to the dressing chain is given by the tuple $(w_0,w_1,w_2,w_3,w_4|a_0,a_1,a_2,a_3,a_4)$, where $a_i$ and $w_i$ are given by \eqref{eq:HM2w}--\eqref{eq:mu2alpha} as:
\begin{align*}
w_0(z)&=-z+\ddz\Big[\log \H_{M_1}(z) - \log \H_{M_0}(z)\Big],&& a_0=-28\\
w_1(z)&=-z+ \ddz\Big[\log \H_{M_2}(z) - \log \H_{M_1}(z)\Big],&& a_1=16,\\
w_2(z)&=-z+ \ddz\Big[\log \H_{M_3}(z) - \log \H_{M_2}(z)\Big],&& a_2=-4,\\
w_3(z)&=-z+ \ddz\Big[\log \H_{M_4}(z) - \log \H_{M_3}(z)\Big],&& a_3=22\\
w_4(z)&=-z+ \ddz\Big[\log \H_{M_0}(z) - \log \H_{M_4}(z)\Big],&& a_4=-16.
\end{align*}
Finally, Proposition~\ref{prop:Mwcorrespondence} implies that the corresponding rational solution to the $A_4$-\p\ system \eqref{eq:Ansystem} is given by the tuple $(f_0,f_1,f_2,f_3,f_4|\a_0,\a_1,\a_2,\a_3,\a_4)$, where
\begin{align*}
f_0(z)&=\tfrac15z+\ddz\Big[\log \H_{M_2}(\cc{3}z) - \log \H_{M_0}(\cc{3}z)\Big],&& \alpha_0=\tfrac{14}{5},\\
f_1(z)&=\tfrac15z+\ddz\Big[\log \H_{M_3}(\cc{3}z) -\log \H_{M_1}(\cc{3}z)\Big],&& \alpha_1=-\tfrac85,\\
f_2(z)&=\tfrac15z+ \ddz\Big[\log \H_{M_4}(\cc{3}z)-\log \H_{M_2}(\cc{3}z)\Big],&& \alpha_2=\tfrac25,\\
f_3(z)&=\tfrac15z+\ddz\Big[\log \H_{M_0}(\cc{3}z) - \log \H_{M_3}(\cc{3}z)\Big],&& \alpha_3=-\tfrac{11}5,\\
f_4(z)&=\tfrac15z+\ddz\Big[\log \H_{M_1}(\cc{3}z) - \log \H_{M_4}(\cc{3}z)\Big],&& \alpha_4=\tfrac85.
\end{align*}
with $\cc{3}^2=-\tfrac1{10}$.

\end{example}

\subsection{Zeros of the special polynomials in the $A_4$ rational solutions}

The zeros of Okamoto and generalized Hermite polynomials that appear in the rational solutions to $\Pfour$ are known to form very regular patterns in the complex plane, \cite{clarkson2003fourth}. In this section we show the equivalent patterns for their $A_4$ counterparts, which are also very regular but show a richer structure.

Following the notation above, we label a Maya diagram in standard form by specifying its signature as a superscript, and 4 non-negative integers $(n_1,n_2,n_3,n_4)$ that  determine the $k$-block coordinates as specified by Theorem~\ref{thm:A4}. More specifically, we can write the sequence of positive integers that belong to $M$ in the following manner:

\begin{eqnarray*}
 M_+^{(5)}(n_1,n_2,n_3,n_4)&=& \Big\{ n_1+j  \Big\}_{j=0}^{n_2-1}\cup \Big\{ n_1+n_2+n_3+j  \Big\}_{j=0}^{n_4-1}\\
 M_+^{(3,1,1)}(n_1,n_2,n_3,n_4)&=&\Big\{ 3(n_1+j)  \Big\}_{j=0}^{n_2-1}\cup \Big\{ 1+3j  \Big\}_{j=0}^{n_3-1} \cup \Big\{ 2+3j  \Big\}_{j=0}^{n_4-1}\\
  M_+^{(1,1,1,1,1)}(n_1,n_2,n_3,n_4)&=&\Big\{ 1+5j  \Big\}_{j=0}^{n_1-1}\cup \Big\{ 2+5j  \Big\}_{j=0}^{n_2-1}\cup \Big\{ 3+5j  \Big\}_{j=0}^{n_3-1} \cup \Big\{ 4+5j  \Big\}_{j=0}^{n_4-1}
 \end{eqnarray*}
 
 Likewise, we denote by $H^{(s)}(n_1,n_2,n_3,n_4)$ the corresponding Hermite Wronskians for each signature $s$, i.e.
 \begin{eqnarray*}
   H^{(5)}_{n_1,n_2,n_3,n_4}(z)   &=& \Wr\left[ \Big\{ H_{n_1+j} (z) \Big\}_{j=0}^{n_2-1}, \Big\{ H_{n_1+n_2+n_3+j}(z)  \Big\}_{j=0}^{n_4-1}\right], \\
 H^{(3,1,1)}_{n_1,n_2,n_3,n_4}(z)&=& \Wr\left[ \Big\{H_{ 3(n_1+j)} (z) \Big\}_{j=0}^{n_2-1}, \Big\{H_{ 1+3j}(z)  \Big\}_{j=0}^{n_3-1} , \Big\{ H_{2+3j}(z)  \Big\}_{j=0}^{n_4-1}\right],\\
  H^{(1,1,1,1,1)}_{n_1,n_2,n_3,n_4}(z)&=& \Wr\left[ \Big\{ H_{1+5j} (z) \Big\}_{j=0}^{n_1-1}, \Big\{ H_{2+5j} (z) \Big\}_{j=0}^{n_2-1}, \Big\{ H_{3+5j} (z) \Big\}_{j=0}^{n_3-1} , \Big\{ H_{4+5j}(z)  \Big\}_{j=0}^{n_4-1}\right].
 \end{eqnarray*}

\begin{figure}[ht] \caption{Zeros of  Hermite Wronskians  $H^{(5)}_{n_1,n_2,n_3,n_4}(z)$ for different values of  $(n_1,n_2,n_3,n_4)$.}\label{fig:1}
 \[\begin{array}{c@{\quad}c@{\quad}c}
\Anfig{2}{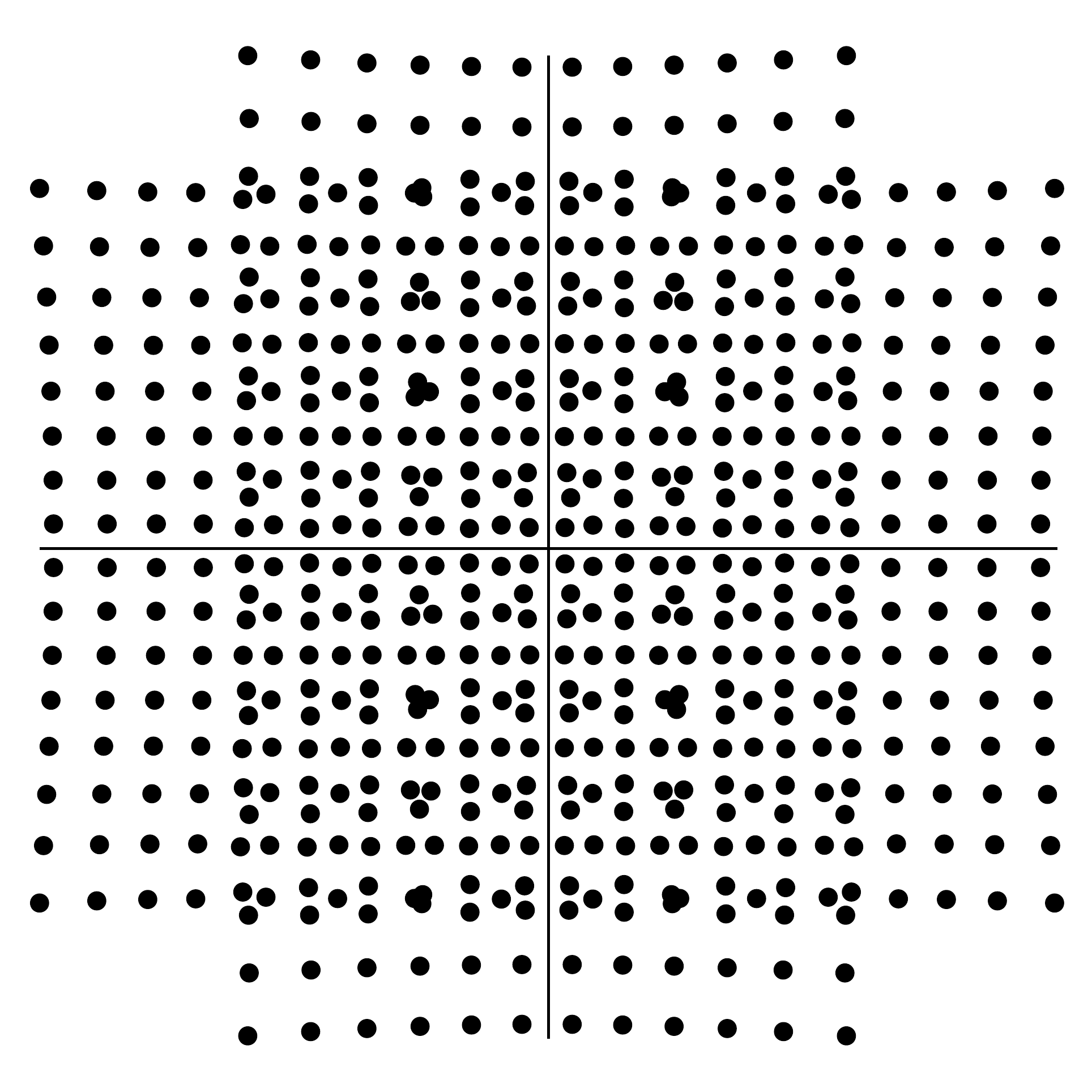} & \Anfig{2}{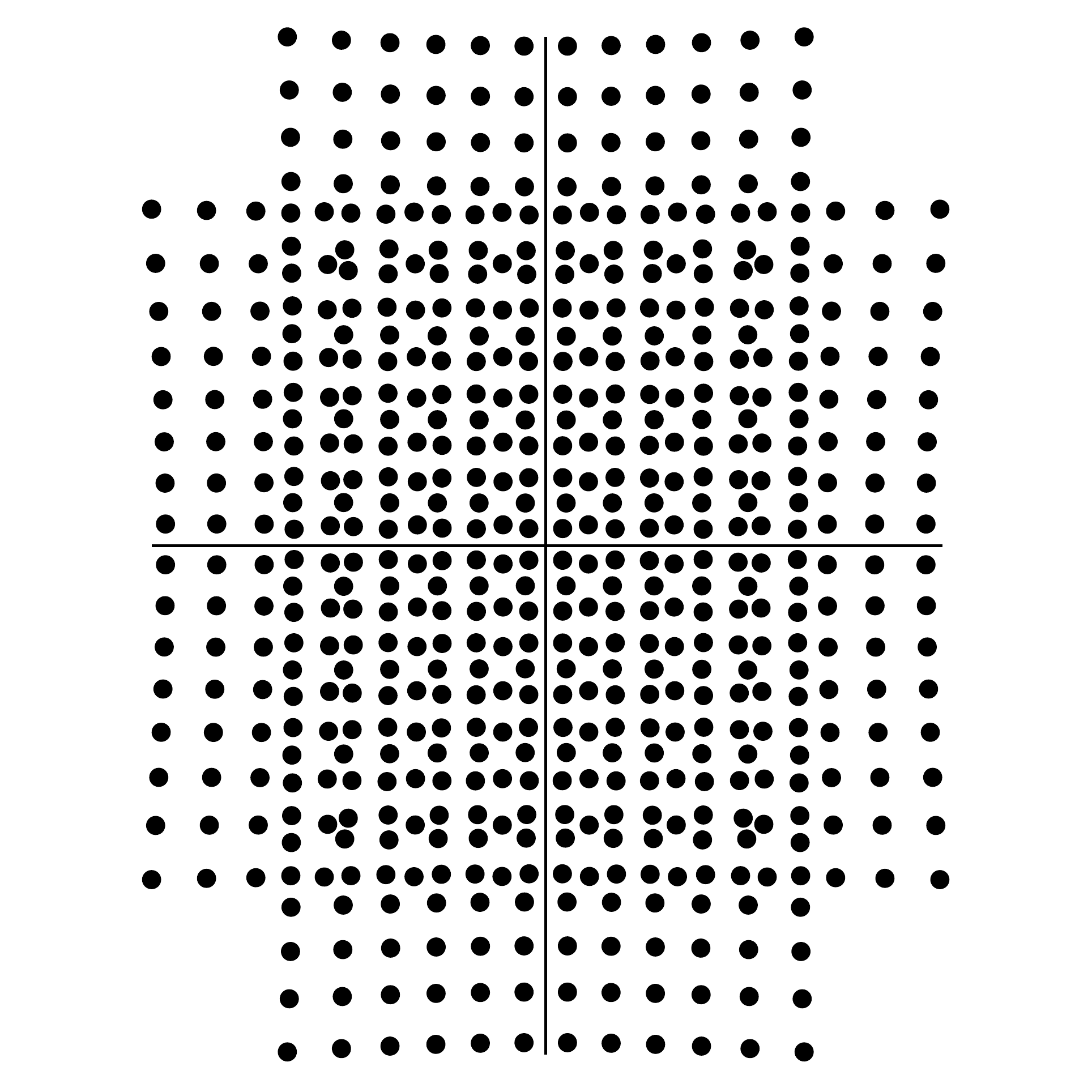} & \Anfig{2}{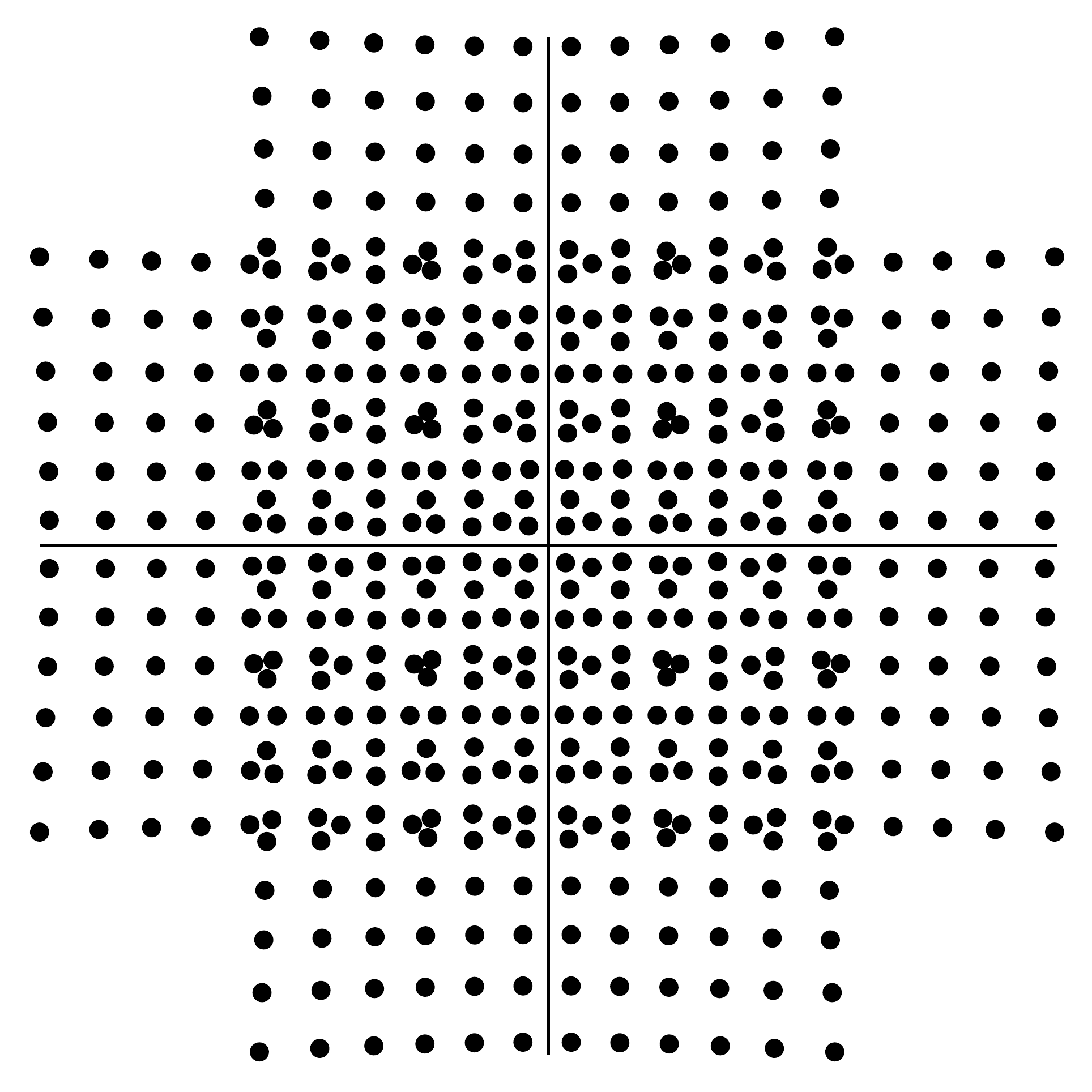} \\ 
H^{(5)}_{12,12,16,16} & H^{(5)}_{12,16,12,16} & H^{(5)}_{12,16,16,12}
\end{array}\]
\end{figure}

\begin{figure}[ht] \caption{Zeros of  Hermite Wronskians  $H^{(3,1,1)}_{n_1,n_2,n_3,n_4}(z)$ for different values of  $(n_1,n_2,n_3,n_4)$.}\label{fig:2}
 \[\begin{array}{c@{\quad}c@{\quad}c}
\Anfig{2}{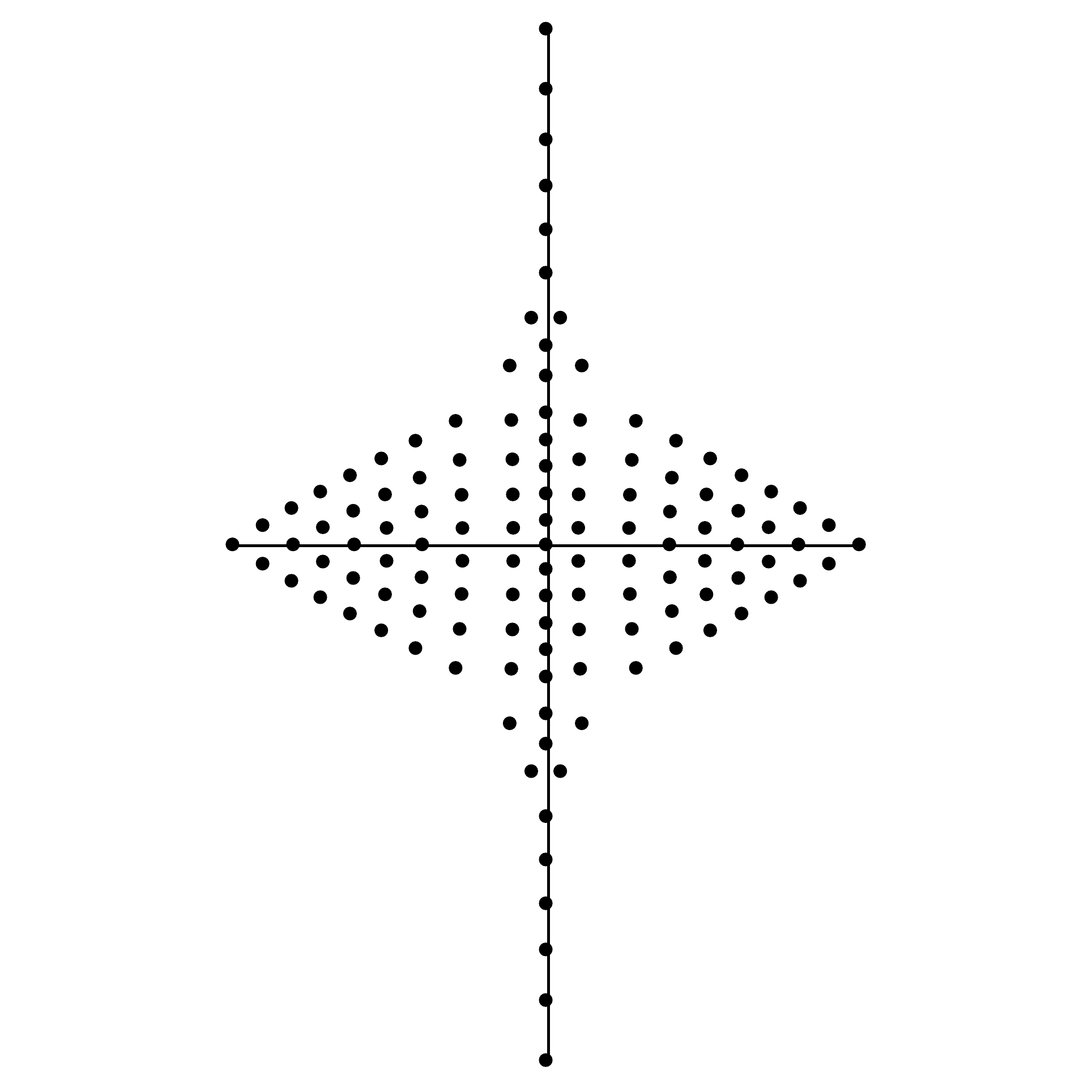} & \Anfig{2}{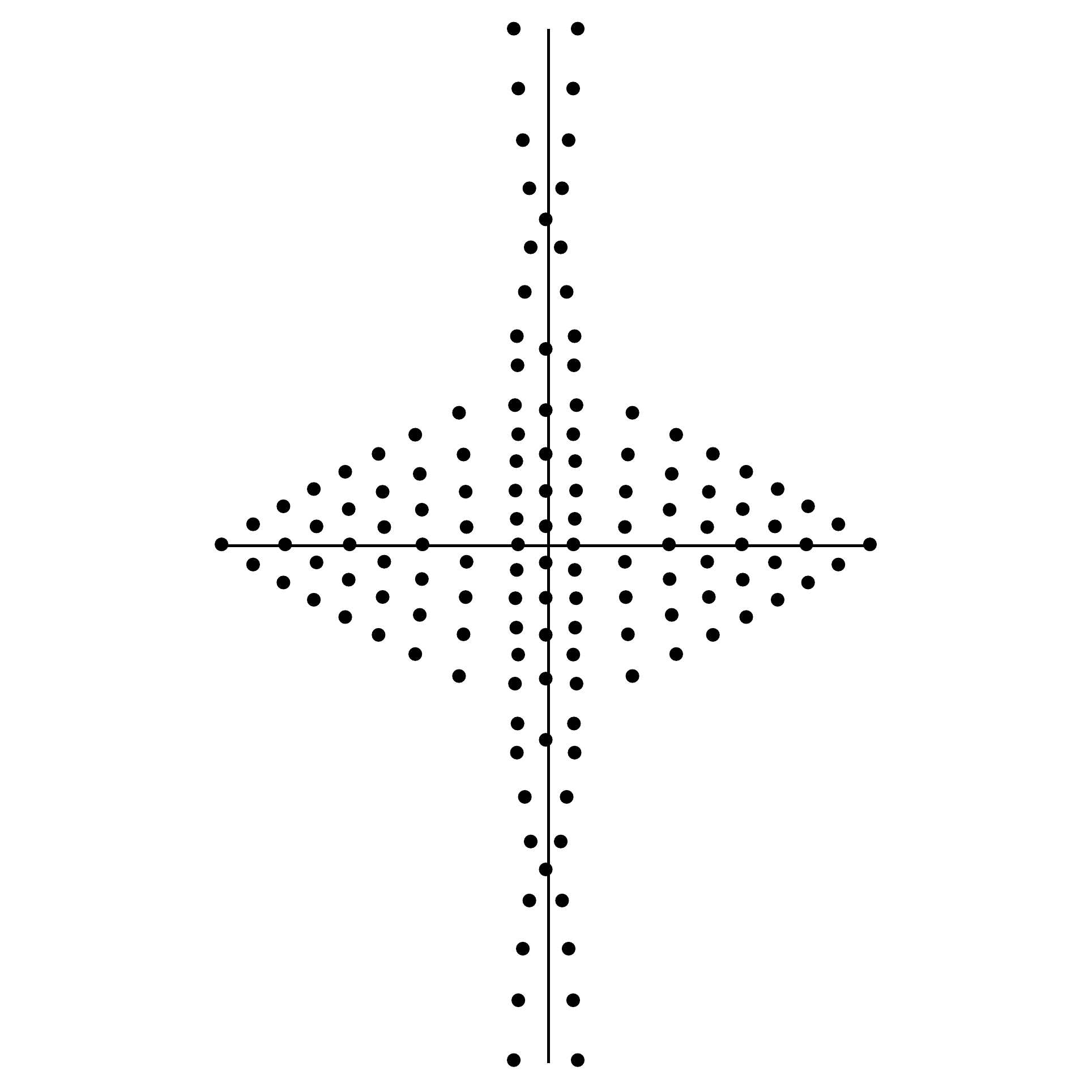} & \Anfig{2}{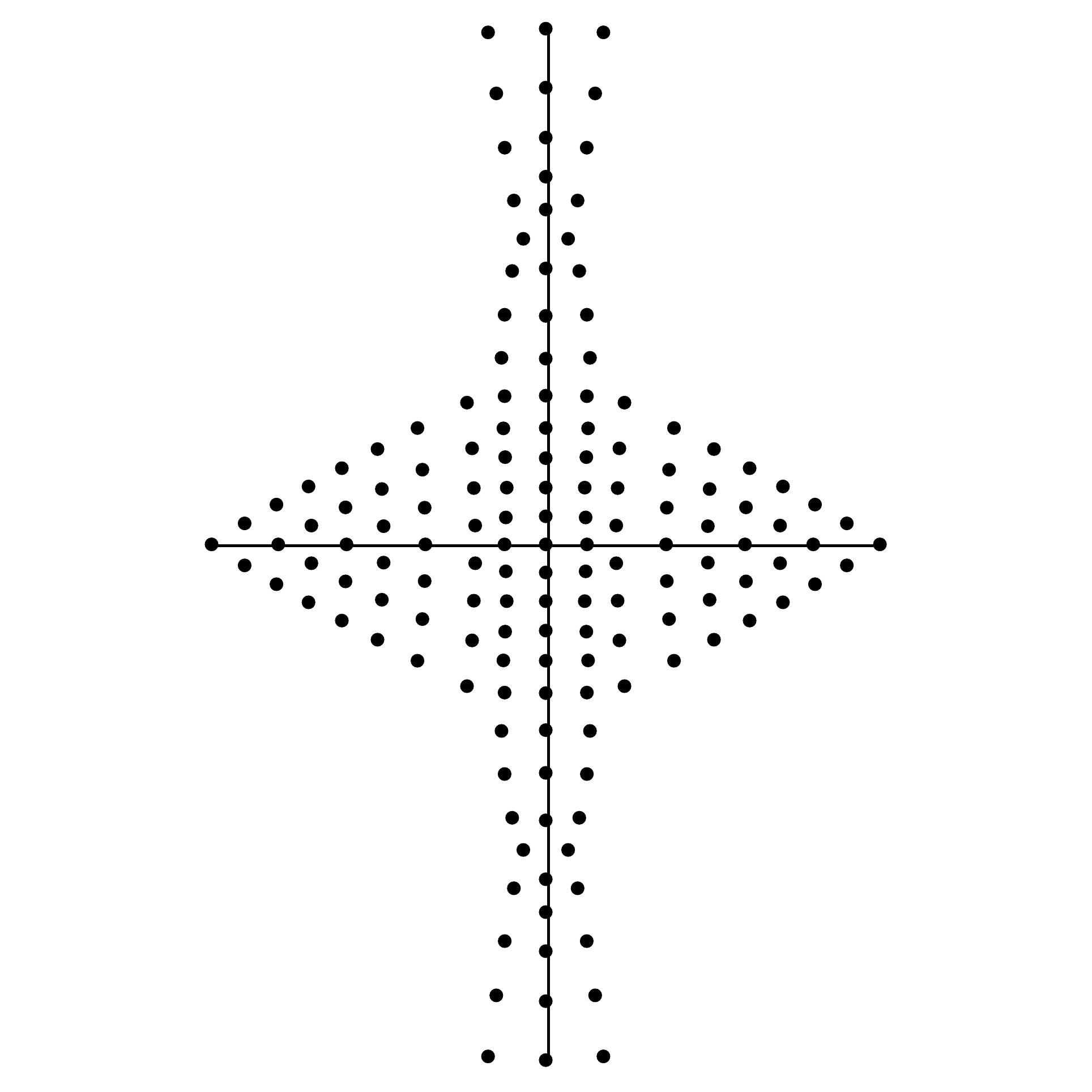} \\ 
H^{(3,1,1)}_{1,5,8,16} & H^{(3,1,1)}_{2,5,8,16} & H^{(3,1,1)}_{3,5,8,16}\\[10pt]
\Anfig{2}{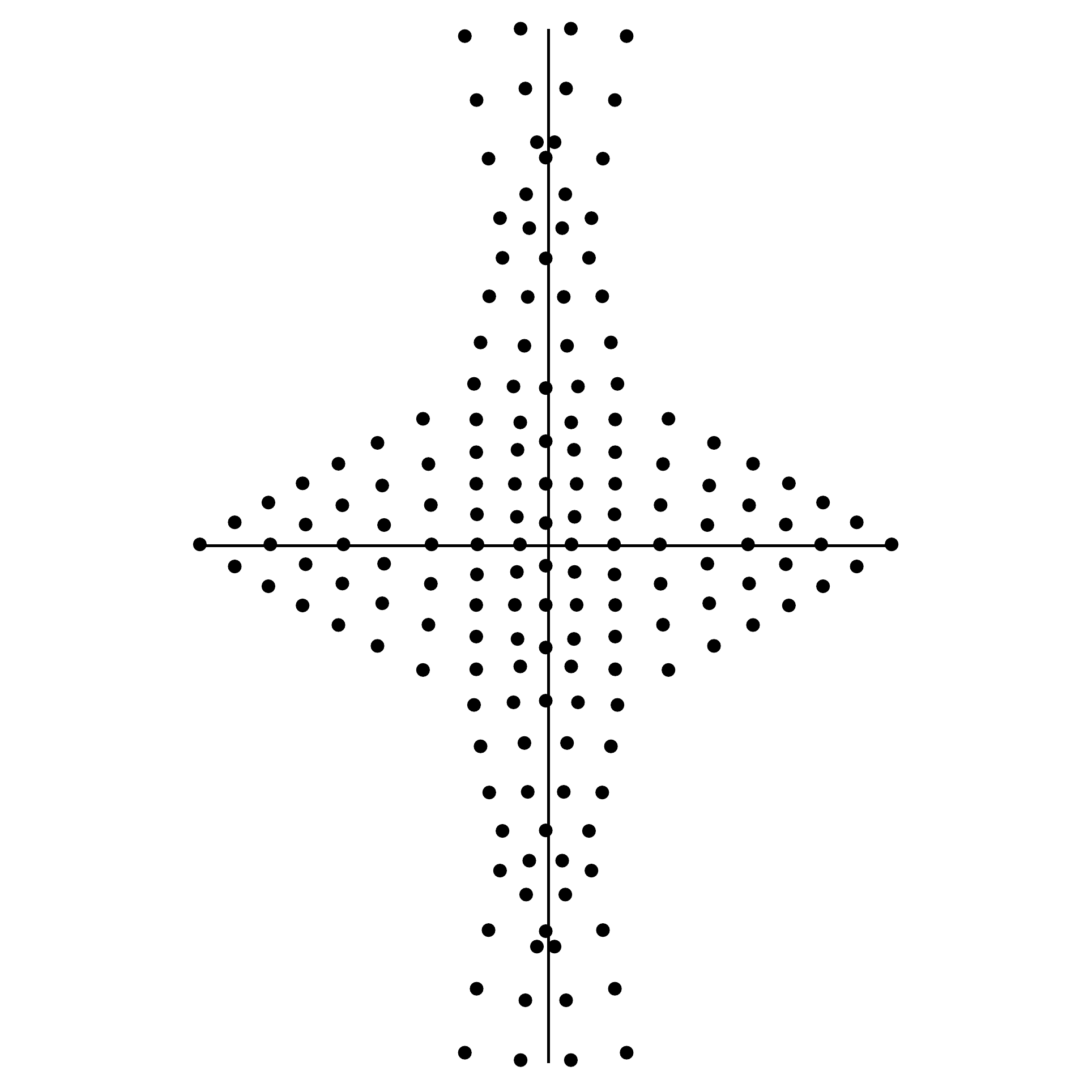} & \Anfig{2}{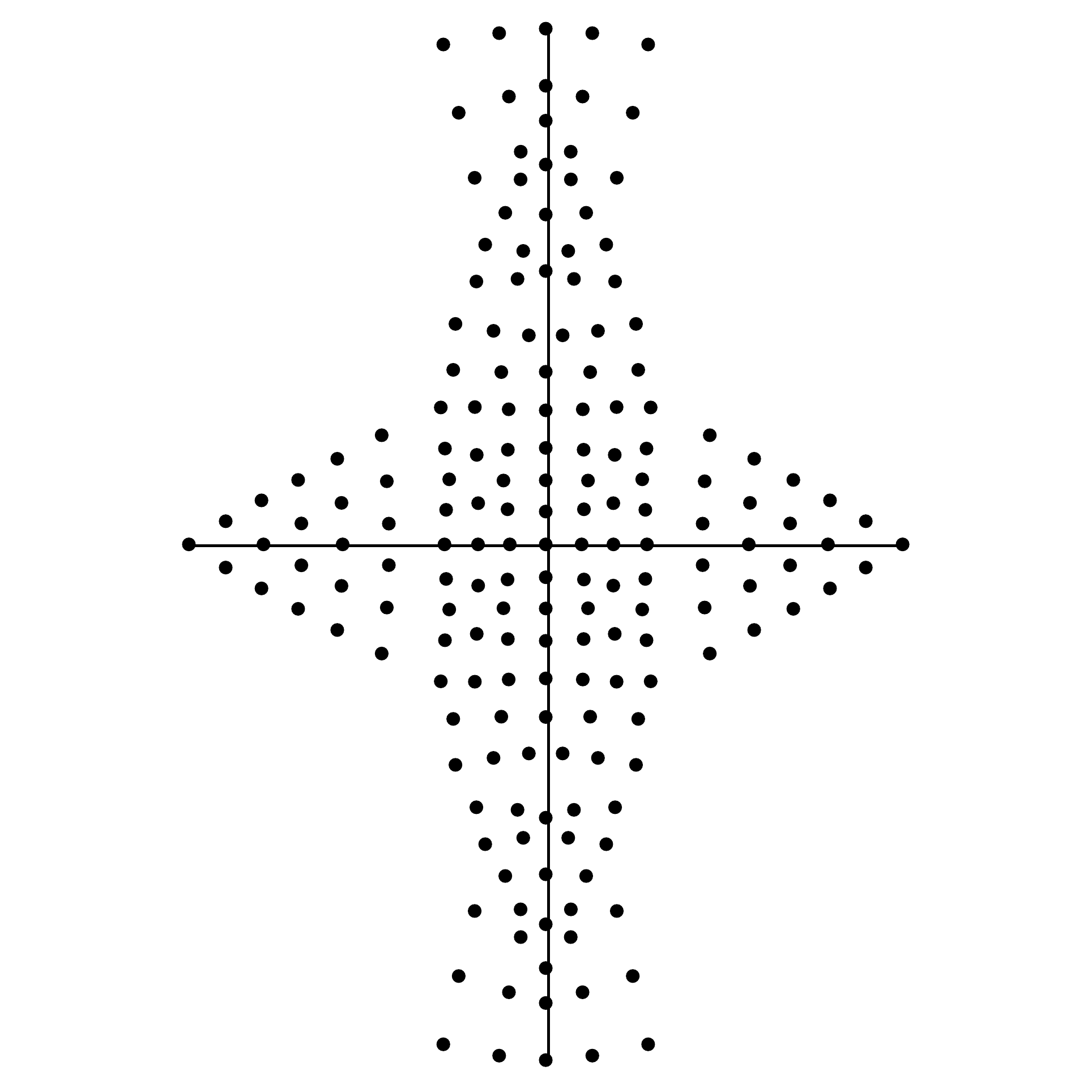} & \Anfig{2}{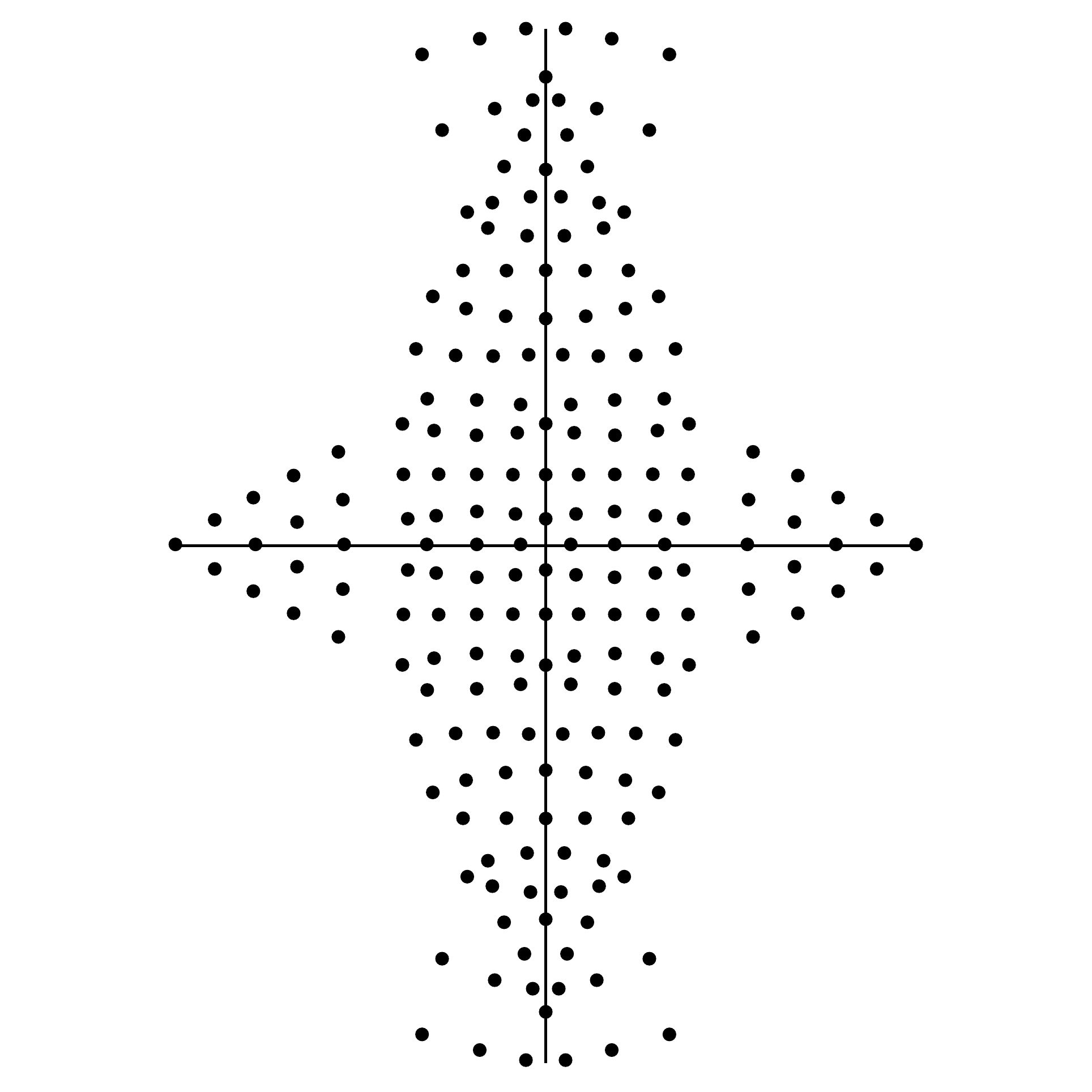} \\ 
H^{(3,1,1)}_{4,5,8,16} & H^{(3,1,1)}_{5,5,8,16} & H^{(3,1,1)}_{6,5,8,16}
\end{array}\]
\end{figure}

\begin{figure}[ht] \caption{Zeros of  Hermite Wronskians  $H^{(1,1,1,1,1)}_{n_1,n_2,n_3,n_4}(z)$ for different values of  $(n_1,n_2,n_3,n_4)$.}\label{fig:3}

\[\begin{array}{c@{\quad}c@{\quad}c}
\Anfig{2}{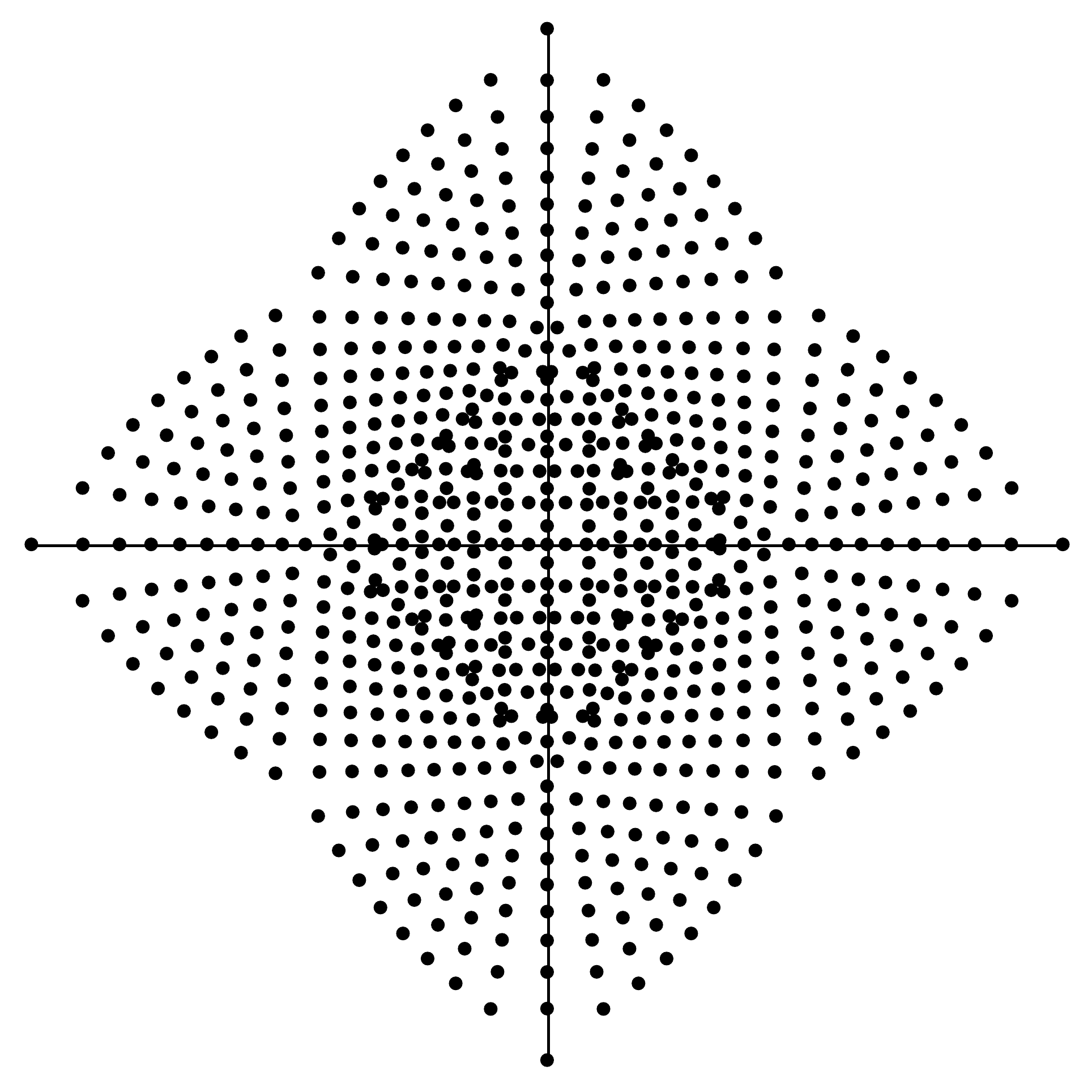}  & \Anfig{2}{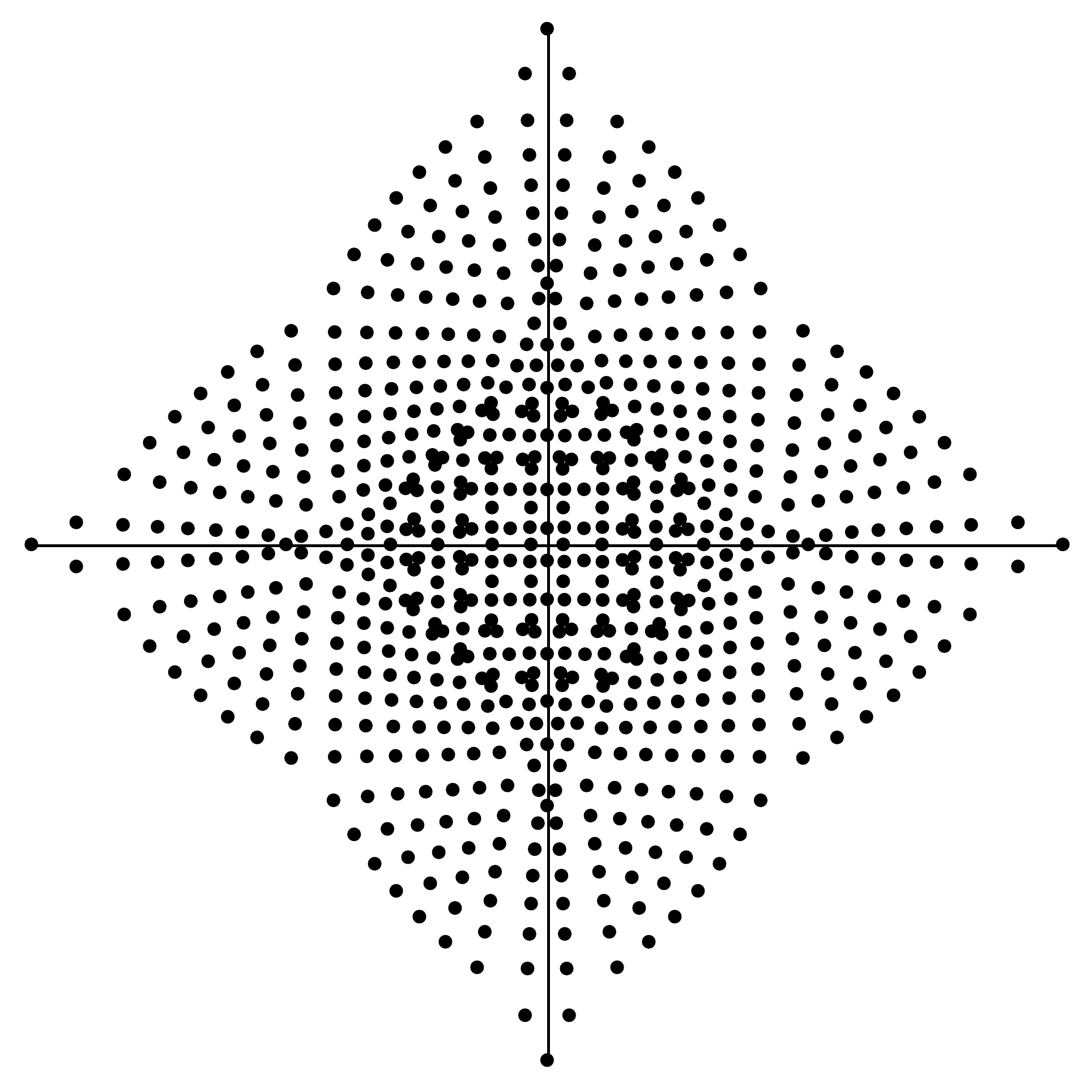} & \Anfig{2}{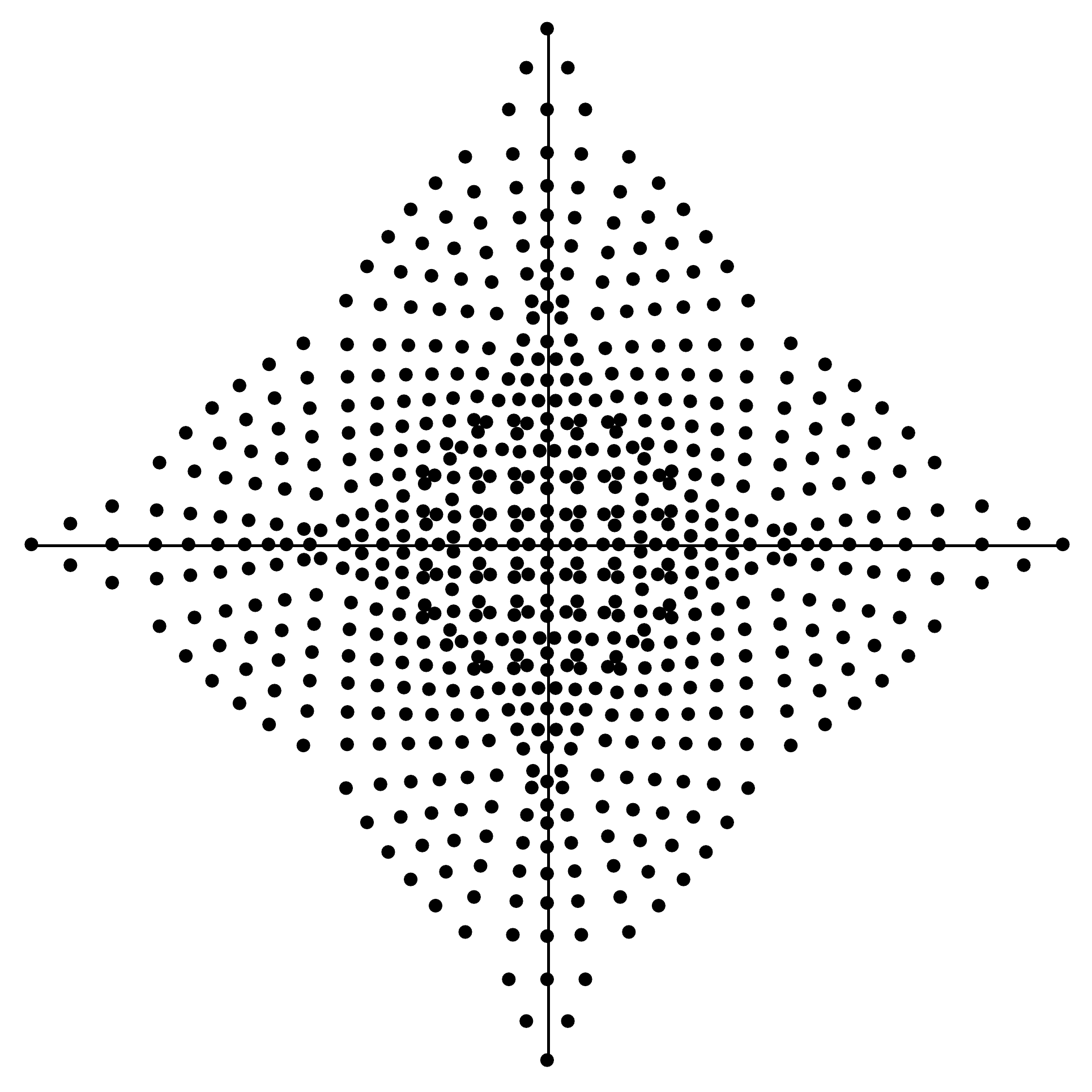}  \\ 
H^{(1,1,1,1,1)}_{1,9,17,18} &H^{(1,1,1,1,1)}_{2,9,16,18} &H^{(1,1,1,1,1)}_{3,9,15,18} \\[10pt]
\Anfig{2}{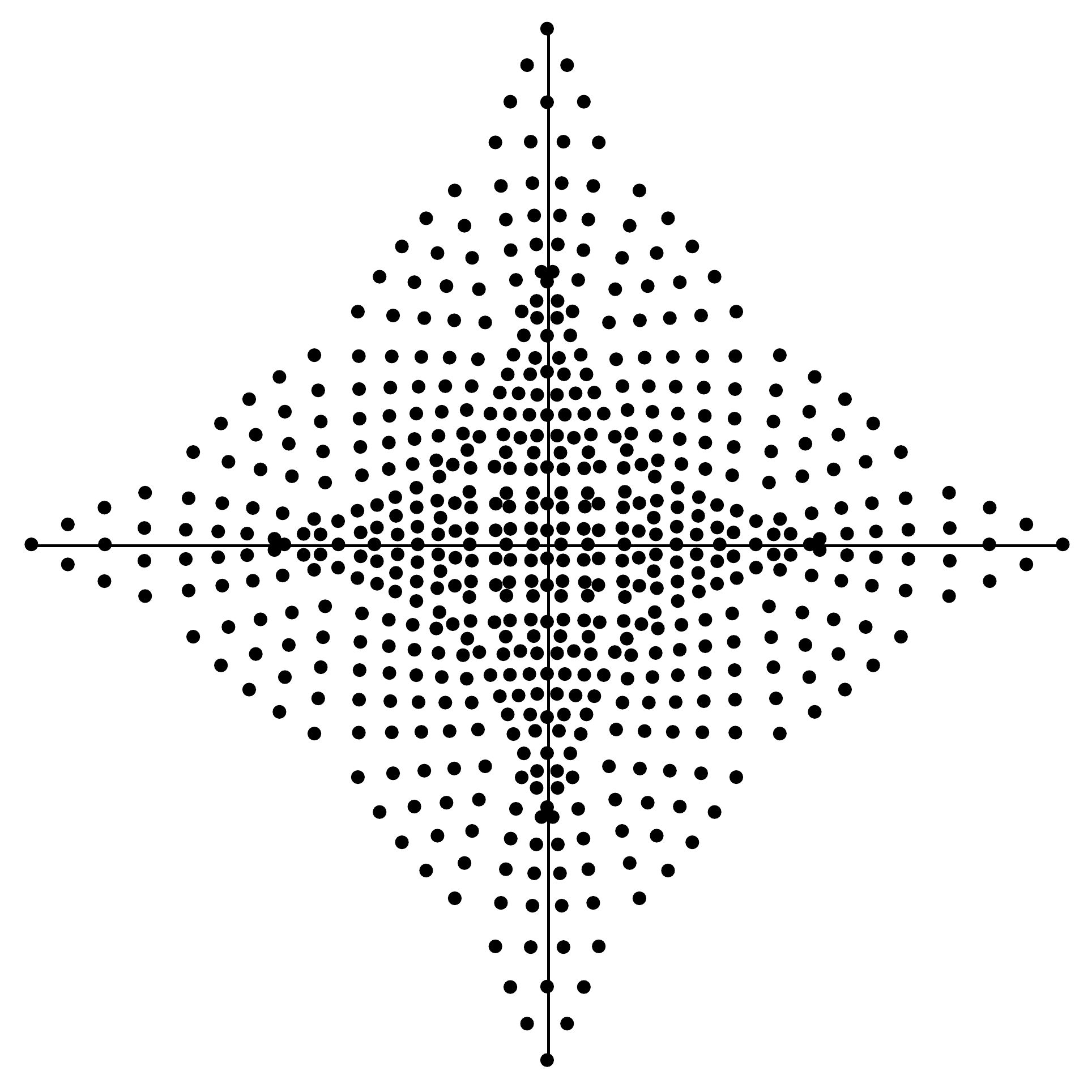} & \Anfig{2}{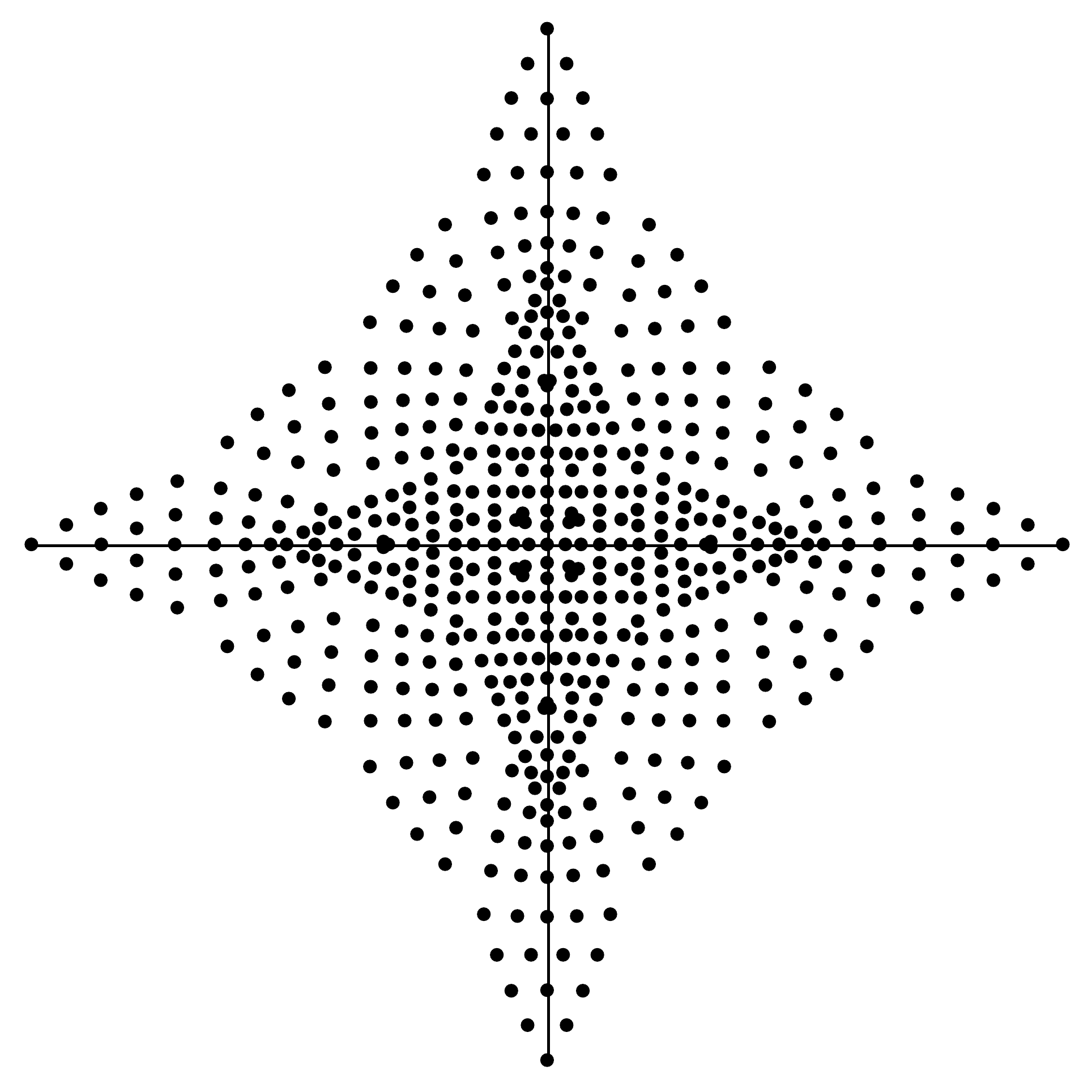}  & \Anfig{2}{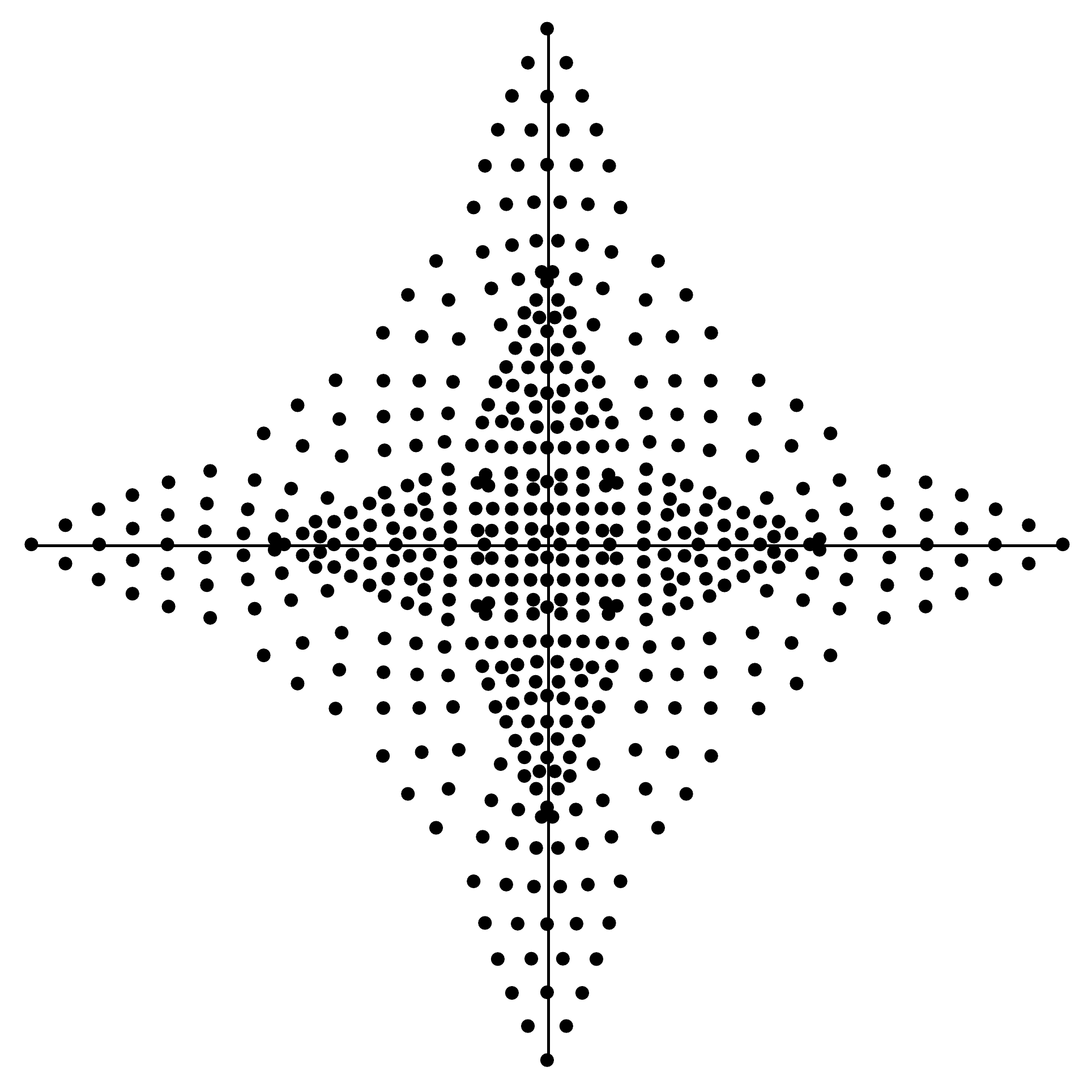}  \\
H^{(1,1,1,1,1)}_{4,9,14,18} &H^{(1,1,1,1,1)}_{5,9,13,18} & H^{(1,1,1,1,1)}_{6,9,12,18} 
\end{array}\]
\end{figure}

From the observation of these plots it is clear that the geometric distribution of the zeros on the complex plane follows some regular patterns that call for an explanation. In some cases, specially in Figure~\ref{fig:1} we observe two overlapping patterns that seems to suggest an approximate factorization of $H^{(5)}_{n_1,n_2,n_3,n_4}(z)$ into the product of two generalized Hermite Wronskians $H^{(3)}_{m_1,m_2}(z)$ and  $H^{(3)}_{m_3,m_4}(z)$. More generally, the approximate correspondence between the Young diagrams of the partitions that determine the sequence of Hermite polynomials in the Wronskian, and the position of the zeros in the complex plane was observed in \cite{FelderHemeryVeselov}. A detailed study of the zeros of these Hermite Wronskians is currently under investigation and we shall not pursue it further here.

We stress that all of these zeros are conjectured by Veselov to be simple \cite{FelderHemeryVeselov}, except the zero at the origin whose multiplicity is a triangular number, \cite{bonneux2019coefficients}. The number of zeros in the real line, and thus the real poles of the rational solutions can be calculated as a function of $(n_1,n_2,n_3,n_4)$ for each signature class by applying the formulas derived in \cite{garcia2015oscillation}. 
\section{Acknowledgements}

The research of DGU has been supported in part by Spanish MINECO-FEDER
Grants MTM2015-65888-C4-3 and PGC2018-096504-B-C33, and by the ICMAT-Severo Ochoa project
SEV-2015-0554. The research of RM was supported in part by NSERC grant
RGPIN-228057-2009. DGU would like to thank the University of Kent for their hospitality during his research stay in Spring 2017 where part of these results were obtained, and the Spanish MINECO program Salvador de Madariaga that provided the necessary financial support. Interesting discussions with Andy Hone, Claire Dunning, Kerstin Jordaan and Galina Filipuk are also gratefully acknowledged. Galina pointed us to the fact that a Wronskian representation for the $k=3$ shift rational solutions in the $A_4$-\p\ case was not known, and this observation  triggered a good amount of the research done in this paper.
PAC would like to thank the Isaac Newton Institute for Mathematical
Sciences for support and hospitality during the programme ``Complex analysis: techniques, applications
and computations'' when some of the work on this paper was undertaken. This work was supported
by  EPSRC grant number EP/R014604/1.

\bibliography{painleve4}

\providecommand{\bysame}{\leavevmode\hbox to3em{\hrulefill}\thinspace}
\providecommand{\MR}{\relax\ifhmode\unskip\space\fi MR }
\providecommand{\MRhref}[2]{%
  \href{http://www.ams.org/mathscinet-getitem?mr=#1}{#2}
}
\providecommand{\href}[2]{#2}
\begin{thebibliography}{10}

\bibitem{adler1978class}
Mark Adler and J{\"u}rgen Moser, \emph{{On a class of polynomials connected
  with the Korteweg-de Vries equation}}, Communications in Mathematical Physics
  \textbf{61} (1978), no.~1, 1--30.

\bibitem{adler1994modification}
V.~\`E. Adler, \emph{{A modification of Crum's method}}, Theoret. and Math.
  Phys. \textbf{101} (1994), no.~3, 1381--1386.

\bibitem{adler1994nonlinear}
\bysame, \emph{{Nonlinear chains and {P}ainlev{\'e} equations}}, Phys. D
  \textbf{73} (1994), no.~4, 335--351.

\bibitem{andrews1998theory}
George~E. Andrews, \emph{The theory of partitions}, Cambridge University Press,
  Cambridge, 1998. \MR{1634067}

\bibitem{andrews2004integer}
George~E. Andrews and Kimmo Eriksson, \emph{Integer partitions}, Cambridge
  University Press, Cambridge, 2004. \MR{2122332}

\bibitem{bagchi2015rational}
B.~Bagchi, Y.~Grandati, and C.~Quesne, \emph{Rational extensions of the
  trigonometric {D}arboux-{P}\"{o}schl-{T}eller potential based on
  para-{J}acobi polynomials}, J. Math. Phys. \textbf{56} (2015), no.~6, 062103.

\bibitem{refBP98}
I.~V. Barashenkov and D.~E. Pelinovsky, \emph{Exact vortex solutions of the
  complex sine-{G}ordon theory on the plane}, Phys. Lett. B \textbf{436}
  (1998), no.~1-2, 117--124. \MR{1649103}

\bibitem{refBNRS}
L.~Bass, J.~J.~C. Nimmo, C.~Rogers, and W.~K. Schief, \emph{Electrical
  structures of interfaces: a {P}ainlev\'{e} {II} model}, Proc. R. Soc. Lond.
  Ser. A Math. Phys. Eng. Sci. \textbf{466} (2010), no.~2119, 2117--2136.
  \MR{2652736}

\bibitem{bermudez2012complex}
D~Berm{\'u}dez and Fern{\'a}ndez~D J, \emph{{Complex solutions to the
  {P}ainlev{\'e} IV equation through supersymmetric quantum mechanics}}, AIP
  Conference Proceedings, vol. 1420, AIP, 2012, pp.~47--51.

\bibitem{bermudez2012complexb}
David Berm{\'u}dez, \emph{{Complex SUSY transformations and the {P}ainlev{\'e}
  IV equation}}, SIGMA \textbf{8} (2012), 069.

\bibitem{refBB15}
Marco Bertola and Thomas Bothner, \emph{Zeros of large degree
  {V}orob'ev-{Y}ablonski polynomials via a {H}ankel determinant identity}, Int.
  Math. Res. Not. IMRN (2015), no.~19, 9330--9399. \MR{3431594}

\bibitem{bonneux2019coefficients}
Niels Bonneux, Clare Dunning, and Marco Stevens, \emph{{Coefficients of
  Wronskian Hermite polynomials}}, arXiv preprint arXiv:1909.03874 (2019).

\bibitem{refBuckPIV}
Robert~J. Buckingham, \emph{Large-degree asymptotics of rational
  {P}ainlev\'{e}-{IV} functions associated to generalized {H}ermite
  polynomials}, Int. Math. Res. Not. IMRN, DOI: 10.1093/imrn/rny172.

\bibitem{refBM12}
Robert~J. Buckingham and Peter~D. Miller, \emph{The sine-{G}ordon equation in
  the semiclassical limit: critical behavior near a separatrix}, J. Anal. Math.
  \textbf{118} (2012), no.~2, 397--492. \MR{3000688}

\bibitem{refBM15}
\bysame, \emph{Large-degree asymptotics of rational {P}ainlev\'{e}-{II}
  functions: noncritical behaviour}, Nonlinearity \textbf{27} (2014), no.~10,
  2489--2578. \MR{3265723}

\bibitem{refBM14}
\bysame, \emph{Large-degree asymptotics of rational {P}ainlev\'{e}-{II}
  functions: critical behaviour}, Nonlinearity \textbf{28} (2015), no.~6,
  1539--1596. \MR{3350600}

\bibitem{burchnall1930set}
JL~Burchnall and TW~Chaundy, \emph{A set of differential equations which can be
  solved by polynomials}, Proceedings of the London Mathematical Society
  \textbf{2} (1930), no.~1, 401--414.

\bibitem{refCCBC}
Hongmei Chen, Min Chen, Gordon Blower, and Yang Chen, \emph{Single-user {MIMO}
  system, {P}ainlev\'{e} transcendents, and double scaling}, J. Math. Phys.
  \textbf{58} (2017), no.~12, 123502, 24. \MR{3734945}

\bibitem{chen2006painleve}
Yang Chen and M.~V. Feigin, \emph{Painlev\'{e} {IV} and degenerate {G}aussian
  unitary ensembles}, J. Phys. A \textbf{39} (2006), no.~40, 12381--12393.

\bibitem{clarkson2003painleve}
Peter~A. Clarkson, \emph{{{P}ainlev{\'e} equations --- nonlinear special
  functions}}, J. Comput. Appl. Math. \textbf{153} (2003), no.~1-2, 127--140.

\bibitem{clarkson2003fourth}
\bysame, \emph{{The fourth {P}ainlev{\'e} equation and associated special
  polynomials}}, J. Math. Phys. \textbf{44} (2003), no.~11, 5350--5374.

\bibitem{clarkson2006special}
\bysame, \emph{{Special polynomials associated with rational solutions of the
  defocusing nonlinear Schr{\"o}dinger equation and the fourth {P}ainlev{\'e}
  equation}}, European J. Appl. Math. \textbf{17} (2006), no.~3, 293--322.

\bibitem{refPAC06review}
\bysame, \emph{Special polynomials associated with rational solutions of the
  {P}ainlev\'{e} equations and applications to soliton equations}, Comput.
  Methods Funct. Theory \textbf{6} (2006), no.~2, 329--401. \MR{2291140}

\bibitem{Clarkson2009vortices}
\bysame, \emph{Vortices and polynomials}, Stud. Appl. Math. \textbf{123}
  (2009), no.~1, 37--62. \MR{2538285}

\bibitem{clarkson2014relationship}
Peter~A. Clarkson and Kerstin Jordaan, \emph{{The relationship between
  semiclassical Laguerre polynomials and the fourth {P}ainlev{\'e} equation}},
  Constr. Approx. \textbf{39} (2014), no.~1, 223--254.

\bibitem{refCM03}
Peter~A. Clarkson and Elizabeth~L. Mansfield, \emph{{The second {P}ainlev\'{e}
  equation, its hierarchy and associated special polynomials}}, Nonlinearity
  \textbf{16} (2003), no.~3, R1--R26. \MR{1975781}

\bibitem{refDG86}
J.~J. Duistermaat and F.~A. Gr\"{u}nbaum, \emph{Differential equations in the
  spectral parameter}, Comm. Math. Phys. \textbf{103} (1986), no.~2, 177--240.
  \MR{826863}

\bibitem{FelderHemeryVeselov}
G.~Felder, A.~D. Hemery, and A.~P. Veselov, \emph{Zeros of {W}ronskians of
  {H}ermite polynomials and {Y}oung diagrams}, Phys. D \textbf{241} (2012),
  no.~23-24, 2131--2137. \MR{2998116}

\bibitem{filipuk2008symmetric}
Galina Filipuk and Peter~A. Clarkson, \emph{{The symmetric fourth
  {P}ainlev{\'e} hierarchy and associated special polynomials}}, Stud. Appl.
  Math. \textbf{121} (2008), no.~2, 157--188.

\bibitem{refFIKNbook}
Athanassios~S. Fokas, Alexander~R. Its, Andrei~A. Kapaev, and Victor~Yu.
  Novokshenov, \emph{{Painlev\'{e} transcendents. The Riemann-Hilbert
  approach}}, Mathematical Surveys and Monographs, vol. 128, American
  Mathematical Society, Providence, RI, 2006. \MR{2264522}

\bibitem{refFW01}
P.~J. Forrester and N.~S. Witte, \emph{Application of the {$\tau$}-function
  theory of {P}ainlev\'{e} equations to random matrices: {PIV}, {PII} and the
  {GUE}}, Comm. Math. Phys. \textbf{219} (2001), no.~2, 357--398. \MR{1833807}

\bibitem{garcia2016bochner}
Ma{\'A}ngeles Garc{\'\i}a-Ferrero, David G{\'o}mez-Ullate, and Robert Milson,
  \emph{A bochner type characterization theorem for exceptional orthogonal
  polynomials}, Journal of Mathematical Analysis and Applications \textbf{472}
  (2019), no.~1, 584--626.

\bibitem{garcia2015oscillation}
M{\textordfeminine}{\'A}ngeles Garc{\'\i}a-Ferrero and David G\'{o}mez-Ullate,
  \emph{{Oscillation theorems for the Wronskian of an arbitrary sequence of
  eigenfunctions of Schr{\"o}dinger's equation}}, Lett. Math. Phys.
  \textbf{105} (2015), no.~4, 551--573.

\bibitem{gomez2013rational}
David G\'{o}mez-Ullate, Yves Grandati, and Robert Milson, \emph{{Rational
  extensions of the quantum harmonic oscillator and exceptional Hermite
  polynomials}}, J. Phys. A \textbf{47} (2013), no.~1, 015203.

\bibitem{gomez2016durfee}
\bysame, \emph{{Durfee rectangles and pseudo-Wronskian equivalences for Hermite
  polynomials}}, Stud. Appl. Math. \textbf{141} (2018), no.~4, 596--625.

\bibitem{gomez2004supersymmetry}
David G\'{o}mez-Ullate, Niky Kamran, and Robert Milson, \emph{{Supersymmetry
  and algebraic Darboux transformations}}, J. Phys. A \textbf{37} (2004),
  no.~43, 10065.

\bibitem{gomez2004darboux}
\bysame, \emph{{The Darboux transformation and algebraic deformations of
  shape-invariant potentials}}, J. Phys. A \textbf{37} (2004), no.~5, 1789.

\bibitem{gomez2009extended}
\bysame, \emph{{An extended class of orthogonal polynomials defined by a
  Sturm--Liouville problem}}, J. Math. Anal. Appl. \textbf{359} (2009), no.~1,
  352--367.

\bibitem{gomez2010extension}
\bysame, \emph{{An extension of Bochner's problem: exceptional invariant
  subspaces}}, J. Approx. Theory \textbf{162} (2010), no.~5, 987--1006.

\bibitem{gomez2013conjecture}
\bysame, \emph{A conjecture on exceptional orthogonal polynomials}, Found.
  Comput. Math. \textbf{13} (2013), no.~4, 615--666.

\bibitem{grandati2011solvable}
Yves Grandati, \emph{Solvable rational extensions of the isotonic oscillator},
  Ann. Physics \textbf{326} (2011), no.~8, 2074--2090.

\bibitem{grandati2012multistep}
\bysame, \emph{{Multistep DBT and regular rational extensions of the isotonic
  oscillator}}, Ann. Physics \textbf{327} (2012), no.~10, 2411--2431.

\bibitem{gromak2008painleve}
Valerii~I Gromak, Ilpo Laine, and Shun Shimomura, \emph{{P}ainlev{\'e}
  differential equations in the complex plane}, vol.~28, Walter de Gruyter,
  2008.

\bibitem{kajiwara1996determinant}
Kenji Kajiwara and Yasuhiro Ohta, \emph{{Determinant structure of the rational
  solutions for the {P}ainlev{\'e} II equation}}, J. Math. Phys. \textbf{37}
  (1996), no.~9, 4693--4704.

\bibitem{kajiwara1998determinant}
\bysame, \emph{{Determinant structure of the rational solutions for the
  {P}ainlev{\'e} IV equation}}, J. Phys. A \textbf{31} (1998), no.~10, 2431.

\bibitem{krein1957continuous}
M.~G. Krein, \emph{On a continual analogue of a {C}hristoffel formula from the
  theory of orthogonal polynomials}, Dokl. Akad. Nauk SSSR (N.S.) \textbf{113}
  (1957), 970--973. \MR{0091396}

\bibitem{marquette2013one}
Ian Marquette and Christiane Quesne, \emph{New ladder operators for a rational
  extension of the harmonic oscillator and superintegrability of some
  two-dimensional systems}, J. Math. Phys. \textbf{54} (2013), no.~10, 102102,
  12. \MR{3134580}

\bibitem{marquette2013two}
\bysame, \emph{Two-step rational extensions of the harmonic oscillator:
  exceptional orthogonal polynomials and ladder operators}, J. Phys. A
  \textbf{46} (2013), no.~15, 155201.

\bibitem{marquette2016}
\bysame, \emph{Connection between quantum systems involving the fourth
  {P}ainlev\'{e} transcendent and {$k$}-step rational extensions of the
  harmonic oscillator related to {H}ermite exceptional orthogonal polynomial},
  J. Math. Phys. \textbf{57} (2016), no.~5, 052101, 15. \MR{3501792}

\bibitem{refMR18}
Davide Masoero and Pieter Roffelsen, \emph{Poles of {P}ainlev\'{e} {IV}
  rationals and their distribution}, SIGMA \textbf{14} (2018), Paper No. 002,
  49. \MR{3742702}

\bibitem{masoero2019roots}
\bysame, \emph{{Roots of generalised Hermite polynomials when both parameters
  are large}}, arXiv preprint arXiv:1907.08552 (2019).

\bibitem{MateoNegro}
J.~Mateo and J.~Negro, \emph{Third-order differential ladder operators and
  supersymmetric quantum mechanics}, J. Phys. A \textbf{41} (2008), no.~4,
  045204, 28. \MR{2451071}

\bibitem{matsuda2012rational}
Kazuhide Matsuda, \emph{{Rational solutions of the Noumi and Yamada system of
  type $A_4^{(1)}$}}, J. Math. Phys. \textbf{53} (2012), no.~2, 023504.

\bibitem{noumi2004painleve}
Masatoshi Noumi, \emph{{P}ainlev{\'e} equations through symmetry}, vol. 223,
  Springer Science \& Business, 2004.

\bibitem{noumi1999symmetries}
Masatoshi Noumi and Yasuhiko Yamada, \emph{{Symmetries in the fourth
  {P}ainlev{\'e} equation and Okamoto polynomials}}, Nagoya Math. J.
  \textbf{153} (1999), 53--86.

\bibitem{NovokShchel14}
V.~Yu. Novokshenov and A.~A. Shchelkonogov, \emph{Double scaling limit in the
  {P}ainlev\'{e} {IV} equation and asymptotics of the {O}kamoto polynomials},
  Spectral theory and differential equations, Amer. Math. Soc. Transl. Ser. 2,
  vol. 233, Amer. Math. Soc., Providence, RI, 2014, pp.~199--210. \MR{3307781}

\bibitem{Novokshenov18}
Victor~Yu. Novokshenov, \emph{Generalized {H}ermite polynomials and
  monodromy-free {S}chr\"{o}dinger operators}, SIGMA \textbf{14} (2018), 106,
  13 pages. \MR{3859422}

\bibitem{oblomkov1999monodromy}
A.~A. Oblomkov, \emph{{Monodromy-free Schr{\"o}dinger operators with
  quadratically increasing potentials}}, Theoret. and Math. Phys. \textbf{121}
  (1999), no.~3, 1574--1584.

\bibitem{odake2011exactly}
Satoru Odake and Ryu Sasaki, \emph{Exactly solvable quantum mechanics and
  infinite families of multi-indexed orthogonal polynomials}, Phys. Lett. B
  \textbf{702} (2011), no.~2-3, 164--170.

\bibitem{odake2013extensions}
\bysame, \emph{Extensions of solvable potentials with finitely many discrete
  eigenstates}, J. Phys. A \textbf{46} (2013), no.~23, 235205.

\bibitem{odake2013krein}
\bysame, \emph{{Krein--Adler transformations for shape-invariant potentials and
  pseudo virtual states}}, J. Phys. A \textbf{46} (2013), no.~24, 245201.

\bibitem{okamoto1987studies3}
Kazuo Okamoto, \emph{Studies on the {P}ainlev\'{e} equations. {III}. {S}econd
  and fourth {P}ainlev\'{e} equations, {$P_{{\rm II}}$} and {$P_{{\rm IV}}$}},
  Math. Ann. \textbf{275} (1986), no.~2, 221--255. \MR{854008}

\bibitem{olsson1994combinatorics}
J{\o}rn~B{\o}rling Olsson, \emph{Combinatorics and representations of finite
  groups}, Fachbereich Mathematik [Lecture Notes in Mathematics], vol.~20,
  Universit{\"a}t Essen, 1994.

\bibitem{refOB05}
N.~Olver and I.~V. Barashenkov, \emph{The complex sine-{G}ordon-2 equation: a
  new algorithm for obtaining multivortex solution on the plane}, Theoret. and
  Math. Phys. \textbf{144} (2005), no.~2, 1223--1226.

\bibitem{refRogers15}
Colin Rogers, \emph{Moving boundary problems for the {H}arry {D}ym equation and
  its reciprocal associates}, Z. Angew. Math. Phys. \textbf{66} (2015), no.~6,
  3205--3220. \MR{3428461}

\bibitem{refRogers16}
\bysame, \emph{On a class of moving boundary problems for the potential mkd{V}
  equation: conjugation of {B}\"{a}cklund and reciprocal transformations}, Ric.
  Mat. \textbf{65} (2016), no.~2, 563--577. \MR{3572529}

\bibitem{refRogers17}
\bysame, \emph{Moving boundary problems for an extended {D}ym equation.
  {R}eciprocal connection}, Meccanica \textbf{52} (2017), no.~15, 3531--3540.
  \MR{3719648}

\bibitem{refRC17}
Colin Rogers and Peter~A. Clarkson, \emph{Ermakov-{P}ainlev\'{e} {II} symmetry
  reduction of a {K}orteweg capillarity system}, SIGMA \textbf{13} (2017),
  Paper No. 018, 20. \MR{3626297}

\bibitem{refRC18}
\bysame, \emph{Ermakov-{P}ainlev\'{e} {II} reduction in cold plasma physics.
  {A}pplication of a {B}\"{a}cklund transformation}, J. Nonlinear Math. Phys.
  \textbf{25} (2018), no.~2, 247--261. \MR{3776560}

\bibitem{sen2005darboux}
Amit Sen, Andrew N.~W. Hone, and Peter~A. Clarkson, \emph{{Darboux
  transformations and the symmetric fourth {P}ainlev{\'e} equation}}, J. Phys.
  A \textbf{38} (2005), no.~45, 9751--9764.

\bibitem{takasaki2003spectral}
Kanehisa Takasaki, \emph{{Spectral curve, Darboux coordinates and Hamiltonian
  structure of periodic dressing chains}}, Comm. Math. Phys. \textbf{241}
  (2003), no.~1, 111--142.

\bibitem{tsuda2005universal}
Teruhisa Tsuda, \emph{{Universal characters, integrable chains and the
  {P}ainlev{\'e} equations}}, Adv. Math. \textbf{197} (2005), no.~2, 587--606.

\bibitem{refUmemura01}
Hiroshi Umemura, \emph{{P}ainlev\'{e} equations in the past 100 years}, A.M.S.
  Translations \textbf{204} (2001), 81--110.

\bibitem{refWVAbook}
Walter Van~Assche, \emph{Orthogonal polynomials and {P}ainlev\'{e} equations},
  Australian Mathematical Society Lecture Series, vol.~27, Cambridge University
  Press, Cambridge, 2018. \MR{3729446}

\bibitem{veselov1993dressing}
A.~P. Veselov and A.~B. Shabat, \emph{{Dressing chains and the spectral theory
  of the Schr{\"o}dinger operator}}, Funct. Anal. Appl. \textbf{27} (1993),
  no.~2, 81--96.

\bibitem{vorob1965rational}
A.~P. Vorob'ev, \emph{{On the rational solutions of the second {P}ainlev{\'e}
  equation}}, Diff. Eqns. \textbf{1} (1965), no.~1, 58--59.

\bibitem{WilloxHietarinta}
Ralph Willox and Jarmo Hietarinta, \emph{Painlev\'{e} equations from {D}arboux
  chains. {I}. {$P_{\rm III}$}--{$P_{\rm V}$}}, J. Phys. A \textbf{36} (2003),
  no.~42, 10615--10635. \MR{2024916}

\bibitem{yablonskii1959rational}
A.~I. Yablonskii, \emph{{On rational solutions of the second {P}ainlev{\'e}
  equation}}, Vesti Akad. Navuk. BSSR Ser. Fiz. Tkh. Nauk. \textbf{3} (1959),
  30--35.

\end{thebibliography}
\bibliographystyle{amsplain}

\end{document}